\documentclass[10pt,twocolumn,twoside]{IEEEtran} 

\usepackage{graphicx}
\usepackage{subcaption}
\usepackage{caption}
\usepackage{amsmath}
\usepackage{amssymb}
\usepackage{amsthm} 
\usepackage{accents}
\usepackage{statex}
\usepackage{cases}
\usepackage{mathrsfs}
\usepackage{tikz,pgfplots}
\pgfplotsset{compat=newest}
\usepackage{xcolor}
\usepackage{caption}
\usepackage{setspace}
\usepackage{booktabs}
\usepackage{array}
\usepackage{cite}
\usepackage{enumitem}
\usepackage{algorithm}
\usepackage{algorithmic}
\usepackage{multirow}
\usetikzlibrary{patterns}
\usetikzlibrary{graphs}
\usetikzlibrary{graphs.standard}
\usetikzlibrary{matrix,arrows,fit,backgrounds,mindmap,plotmarks,decorations.pathreplacing}
\newcolumntype{C}[1]{>{\centering\arraybackslash}p{#1}}

\usetikzlibrary{automata}

\newtheorem{theorem}{Theorem}

\newtheorem{lemma}{Lemma}

\newtheorem{remark}{Remark}
\newtheorem{assumption}{Assumption}

\newlength\figureheight
\newlength\figurewidth

\DeclareMathOperator{\1}{\textbf{1}} 
\DeclareMathOperator{\0}{\textbf{0}} 

\DeclareMathOperator*{\col}{col}
\DeclareMathOperator*{\tr}{tr}
\DeclareMathOperator*{\cov}{cov}
\DeclareMathOperator*{\rank}{rank}

\tikzstyle{sensor} = [draw, fill=blue!20, rectangle, rounded corners,
minimum height=2em, minimum width=7em]
\tikzstyle{est} = [draw, fill=orange!20, rectangle, rounded corners,
minimum height=2em, minimum width=7em]
\tikzstyle{pinstyle} = [pin edge={to-,thin,black}]

\title{A Framework for Distributed Estimation with Reduced Communication via Event-Based Strategies}

\author{Jiaqi Yan, Yilin Mo, and Hideaki Ishii
	\thanks{
		J. Yan and H. Ishii are with the Department of Computer Science, Tokyo Institute of Technology, Japan. Emails: {jyan@sc.dis.titech.ac.jp, ishii@c.titech.ac.jp}. Yilin Mo is with the Department of Automation, BNRist, Tsinghua University, P. R. China. Email: ylmo@tsinghua.edu.cn.
	}
\thanks{This work was supported in the part by JSPS under Grants-in-Aid for Scientific Research Grant No. 18H01460 and 21F40376.}
}

\begin{document}
\maketitle

\begin{abstract}
In this paper, we consider the problem of distributed estimation in a sensor network, where multiple sensors are deployed to estimate the state of a linear time-invariant (LTI) Gaussian system. By losslessly decomposing the Kalman filter, a framework of event-based distributed estimation is developed, where each sensor node runs a local filter using solely its own measurement, alongside with an event-based synchronization algorithm to fuse the neighboring information. One novelty of the proposed framework is that it decouples the local filter from synchronization process. By doing so, we prove that a general class of triggering strategies can be applied in our framework, which yields stable distributed estimators under the minimal requirements of network connectivity and collective system observability. Moreover, the developed results can be generalized to achieve a distributed implementation of any Luenberger observer. By solving a semi-definite programming (SDP), we further present a low-rank estimator design to obtain the optimal gain of Luenberger observer such that the distributed estimation is realized under the constraint of message size. Therefore, as compared with existing works, the proposed algorithm enjoys lower data size for each transmission. Numerical examples are finally provided to demonstrate the proposed methods.
\end{abstract}


\section{Introduction}
State estimation in a sensor network, as one of the most important focuses in the past couple of decades, has attracted significant research attention due to its wide applications in environment monitoring, target tracking, robotics navigation, 
\textit{etc.} (see \cite{subbotin2009design,xie2012fully,luo2005universal,vu2015distributed,jia2016cooperative,yan2021resilient} for examples). However, the classical centralized estimators are no longer suitable in many networks where the data size increases rapidly and collecting all information in a data center becomes difficult. As such, distributed estimation algorithms are required where every sensor produces local estimates using its own measurements and information exchange with only immediate neighbors.

Within this field, a fundamental problem is to estimate the state of a linear time-invariant (LTI) Gaussian system by using multiple sensors. Obviously, the optimal solution of it is provided by 
the centralized Kalman filter (\hspace{1pt}\cite{anderson2012optimal}). In order to achieve a distributed implementation of Kalman filter, a number of consensus-based distributed estimators have been proposed in the literature including  \cite{bar1986effect,olfati2005distributed,olfati2007distributed,chen2019distributed,olfati2009kalman,battistelli2014consensus,li2011consensus,battistelli2016stability,del2009distributed,kar2010gossip,battistelli2014kullback,farina2010distributed,farina2010distributed}.
For example, by performing average consensus on local estimates, a Kalman-Consensus Filter (KCF) is proposed in \cite{olfati2007distributed}. This work has also inspired a group of distributed estimators where local estimates are fused by consensus algorithms (\hspace{1pt}\cite{chen2019distributed,olfati2009kalman,farina2010distributed}). In contrast, Battistelli \textit{et al.} \cite{battistelli2014consensus} suggest performing consensus algorithms on both measurements and inverse-covariance matrices. The developed estimator can guarantee the stability of estimation error even when the system has nonlinear dynamics. Other popular approaches include \cite{chen2002estimation} and \cite{battistelli2014kullback}, where the distributed estimators are established by performing consensus on information matrices and probability densities, respectively.

In the aforementioned works, the distributed estimators require at least one consensus step during each sampling period. However, the autonomous agents are often equipped with embedded microprocessors, onboard communication and actuation modules that are powered by batteries and thus have limited energy resources. In this respect, it is not surprising that event-triggered data transmission policies have recently been popular for their capabilities of improving the
resource utilization efficiency (\hspace{1pt}\cite{kadowaki2014event,nowzari2019event,dimarogonas2011distributed,garcia2014decentralized,yi2018dynamic,mishra2021dynamic}). Along this line, some event-triggered transmission strategies have been proposed for realizing distributed estimation. Particularly, the information flow of the existing works is illustrated in Fig.~\ref{fig:existing}. Notice that $\Delta_i(k_s^i)$ in the figure represents the message transmitted by sensor $i$. Consensus algorithms are performed on $\Delta_i(k_s^i)$ towards achieving stable distributed estimators. In the literature, the transmission instants of $\Delta_i(k_s^i)$ can be determined by state-based triggering strategies \cite{liu2018event,battistelli2018distributed,yu2020event,wu2016finite}, measurement-based triggering strategies \cite{song2019event,shi2016event}, or innovation-based triggering strategies \cite{liu2015event,yan2014distributed}. It is noticed from the figure that the local filters are usually coupled with consensus/synchronization processes in existing works. As such, the performance of both processes are inevitably affected by the triggering mechanism, bringing more challenges to analyze and guarantee the performance of distributed estimators. Therefore, existing works usually have additional requirements on network topology or system matrix to ensure the stability of local estimates. For example,
Battistelli \textit{et al.}	\cite{battistelli2018distributed} propose a consensus-based distributed Kalman filter, where the triggering instants are determined by both the state estimates and error covariance. Under the assumption that the system matrix is invertible and the communication topology is strongly connected, they prove the mean-square boundedness of the estimation error. Based on a stochastic triggering condition, a minimum mean square error estimator is given in \cite{yu2020event}. Again, if there exist a invertible system matrix and strongly connected communication graph, the distributed estimator is stable with a bounded mean-square estimation error. Other examples include \cite{liu2015event}
and \cite{yan2014distributed}, where the authors respectively develop the event-triggered transmission
strategies based on the distance between the current and latest transmitted innovations. 
Even so, results in this area have been scattered in the literature, which is worth further research efforts.

\begin{figure}
	\centering
\includegraphics[width=0.35\textwidth]{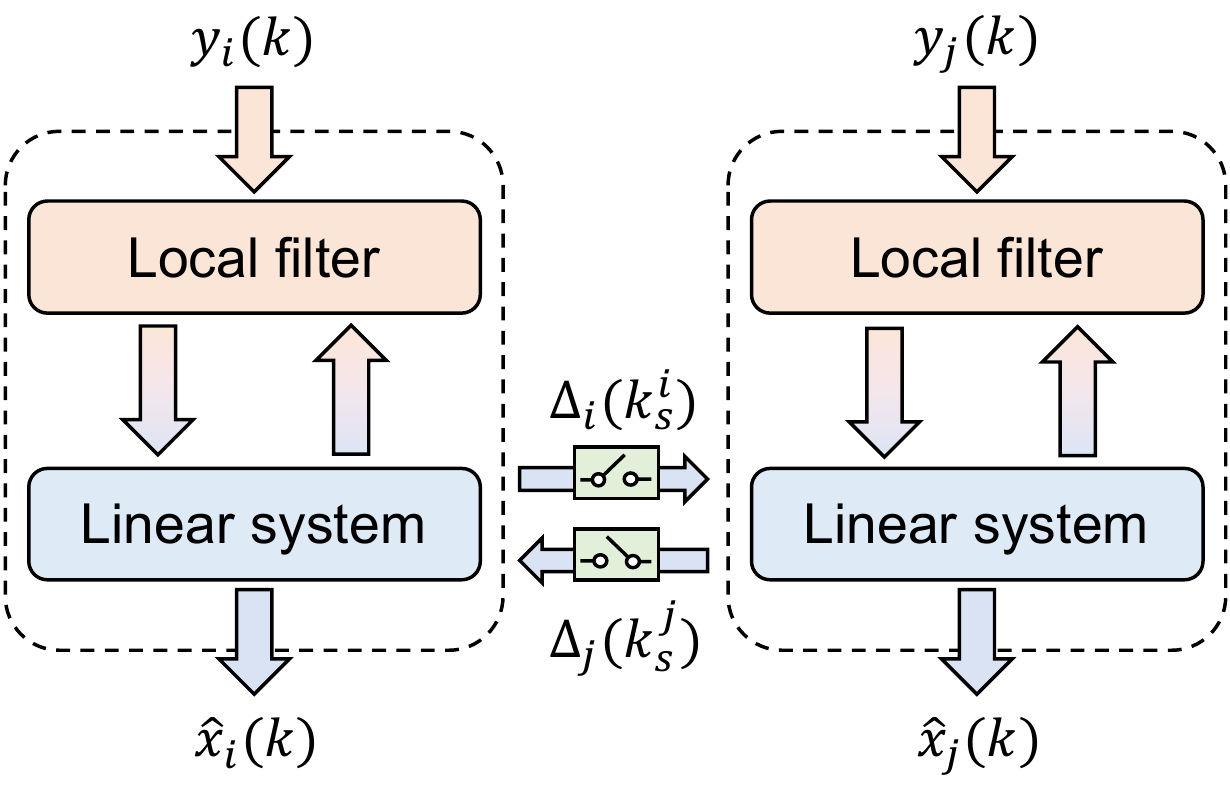}
	\caption{The information flow of most existing event-based distributed estimation algorithms, where sensors $i$ and $j$ are immediate neighbors.}
	\label{fig:existing}
\end{figure}

Inspired by the considerations above, this paper also studies the problem of event-based distributed estimation. We particularly focus on approaching the performance of the steady-state Kalman filter, which is commonly used in practice and has the identical asymptotic estimation performance as the one with time-varying estimation gain \cite{shi2011sensor}. In contrast to the existing solutions, a novel estimation framework will be presented alongside with event-based communication strategies, which decouples the local filters from the consensus process by performing a decomposition of the Kalman filter. The main contributions of this paper are described as follows: 

1) By decomposing the Kalman filter, we prove that the optimal Kalman estimate can be perfectly recovered as a weighted sum of a bank of local filters. Accordingly, this paper presents a novel framework for the event-based distributed implementation of the Kalman filter. Here, each sensor performs local filtering solely with its own measurement based on a decomposition of Kalman filter, and global fusion is realized through information exchange by running an event-based synchronization algorithm. With this framework, we decouple the local filter from the consensus process and reformulate the problem of distributed state estimation into that of synchronization among stochastic linear systems. As a result of decoupling, the performance of local filters will not be affected even when no sensors are triggered to transmit at certain times.

2) Instead of working with any specific triggering function, we show that a general class of triggering strategies can be applied to solve the problem of distributed estimation under conditions on Laplacian and system matrix. The estimation performance is further analyzed, where the presence of noises and event-triggered mechanisms prevent the approaches of Lyapunov stability for deterministic systems from being directly applied. To solve this problem, we extend the classical results on the stability of supermartingale and propose a $c$-martingale convergence lemma. Based on it, the proposed estimator is proved to be stable at each sensor side under the minimal requirements of network connectivity and collective system observability. This extends, in a non-trivial way, the results in our previous work \cite{yan2021distributed} for the full transmission case. 

3)  By running the proposed algorithm, at each triggering instant, the size of message exchanged between paired agents is equal to $\rank(K)\leq \min\{m,n\}$, where $n$ and $m$ are respectively dimensions of the system state and sensor measurement, and $K\in\mathbb{R}^{n\times m}$ is the steady-state Kalman gain. In contrast, due to the coupling of local filters and the consensus process, existing event-based distributed estimators usually require the message exchange with a larger size in order to account for the performance loss on both processes caused by the intermittent communication. For instance, in the works \cite{liu2018event,battistelli2018distributed,yu2020event,wu2016finite,song2019event,shi2016event,liu2015event,yan2014distributed}, the estimation algorithms require exchange of the information on local covariance matrix which is of size $n^2$ at each transmission. Therefore, as compared with these works, our estimator enjoys lower message complexity\footnote{In this paper, message complexity is defined as the size of message transmitted at each triggering instant. Notice that in practice, it usually costs a fixed number of bits ($8$ or $16$ bits) to transfer a real value. Therefore, data rate increases linearly with the message complexity. 
}.  

4) Notice that in practice, a communication channel in the sensor network is usually limited by a finite bandwidth. Therefore, we further investigate the design of distributed estimators under the constraint of message complexity. To this end, it is shown that the framework proposed in this paper can be generalized to achieve a distributed implementation of any Luenberger observer, which may not necessarily be the Kalman filter. Suppose that the message complexity that the network is willing to tolerate is $\tilde{r}$, where $0<\tilde{r} \leq \min\{m,n\}$. This paper, by solving a semi-definite programming, presents how to design the optimal estimation gain of Luenberger observer such that the distributed estimator can be implemented with message complexity no more than $\tilde{r}$.

The remainder of this paper is organized as follows. Section~\ref{sec:formulation} introduces the system settings and presents the problem of distributed estimation. A decomposition of the Kalman filter is introduced in Section~\ref{sec:decompose}, based on which Section~\ref{sec:implement} presents the framework of distributed estimation with an event-triggered communication strategy. The performance of the proposed estimation algorithm is also analyzed in Section~\ref{sec:analysis}. We then discuss how to design the estimation gain under the constraint of message complexity in Section \ref{sec:Kdesign} and validate the algorithm performance through numerical examples in Section \ref{sec:simulation}. Finally, we conclude this work in Section~\ref{sec:conclusion}.

A preliminary version of this paper has been submitted for conference presentation \cite{yan2022necsys}. As compared to it, the current version presents a different decomposition method of the Kalman filter and further proposes the low rank estimator design to reduce the message complexity. Also, we present all the proofs as well as more extensive discussions and numerical examples here.

\textit{Notations}: For a group of vectors $v_{i} \in \mathbb{R}^{m_{i}},$ the vector $\left[v_{1}^{T}, \ldots, v_{N}^{T}\right]^{T}$ is also written as $\col(v_{1}, \ldots, v_{N}).$ We denote by $\rho(A)$ the spectral radius of any matrix $A$. Moreover, given a positive semidefinite matrix $U$, let $U^{1/2}$ be the positive semidefinite matrix that satisfies $U^{1/2}U^{1/2}=U$.

\section{Problem Formulation}\label{sec:formulation}
\subsection{System setup for distributed estimator}
In this paper, we consider the LTI Gaussian system as given below:
\begin{equation}\label{eqn:plant}
x(k+1) = Ax(k) + w(k),
\end{equation}
where $x(k)\in\mathbb{R}^n$ is the system state to be estimated, $w(k)\sim \mathcal{N} (0, Q )$ is an independent and identically distributed (i.i.d.) Gaussian noise with zero mean and covariance matrix $Q\geq 0$. Moreover, the initial state $x(0)$ also follows the Gaussian distribution which has zero mean. 

A sensor network monitors the system above, where the measurement from each sensor $i\in\{1,2,...,m\}$ is given by\footnote{If sensor $i$ outputs a vector measurement, we can treat each of the components as a scalar measurement.}
\begin{equation}\label{eqn:sensoroutput}
y_i(k) = C_ix(k) + v_i(k),
\end{equation}
where $y_i(k)\in\mathbb{R}$ is the measurement of sensor $i$ and $C_i\in\mathbb{R}^{1\times n}$.

By collecting the measurements from all sensors, we have
\begin{equation}\label{eqn:sensormatrix}
y(k) = Cx(k) + v(k),
\end{equation}
where
\begin{equation*}
\begin{split}
y(k) \triangleq {\left[\begin{array}{c}
	y_{1}(k) \\
	\vdots \\
	y_{m}(k)
	\end{array}\right],} \; C \triangleq {\left[\begin{array}{c}
	C_{1} \\
	\vdots \\
	C_{m}
	\end{array}\right],} \;
v(k) \triangleq {\left[\begin{array}{c}
	v_{1}(k) \\
	\vdots \\
	v_{m}(k)
	\end{array}\right]},
\end{split}
\end{equation*}
and $v(k)$ is a zero-mean i.i.d. Gaussian noise with covariance $R\geq0$ and is independent of $w(k)$ and $x(0)$. 

The system \eqref{eqn:plant} need not be stable, but throughout this paper, we make the following assumption on system observability:
\begin{assumption}[Collective observability]\label{assup:observable}
	The system is collectively observable, i.e., the pair $(A, C)$ is
	observable, while $(A, C_i)$ is not necessarily observable for each sensor $i\in\{1,\cdots,m\}$.
\end{assumption}

Notice that Assumption~\ref{assup:observable} requires that the measurements from all sensors jointly guarantee the system observability, while for a single sensor, it may not be able to observe the whole state space.

In this paper, we aim to design a distributed algorithm to estimate the system state by the sensor network, which operates over the communication topology modeled by a connected undirected graph $\mathcal{G}=(\mathcal{V},\mathcal{E})$. Here, $\mathcal{V} =\{1,2,...,m\}$ and $\mathcal{E}\subset \mathcal{V}\times\mathcal{V}$ are the set of sensors and edges, respectively. Moreover, the interaction among sensors is described by the weighted adjacency matrix $\mathcal{A}=\left[a_{i j}\right]$, where $a_{ij}\geq 0$ and $a_{ij}=a_{ji},\forall i,j \in \mathcal{V}$. Notice that $(i,j)\in \mathcal{E}$ if and only if $a_{ij}>0$. The degree matrix of  $\mathcal{G}$ is defined as $\mathcal{D}_{\mathcal{G}} \triangleq \diag\left(d_{1}, \ldots, d_{m}\right)$, where $d_{i}=\sum_{j=1}^{m} a_{ij}$. The Laplacian matrix of $\mathcal{G}$ is calculated as $\mathcal{L}_{\mathcal{G}}\triangleq\mathcal{D}_{\mathcal{G}}-\mathcal{A}$. Since $\mathcal{G}$ is connected, let us arrange the eigenvalues of $\mathcal{L}_{\mathcal{G}}$ as 
\begin{equation}\label{eqn:mu}
	0=\mu_1< \mu_2 \leq \cdots \leq \mu_m.
\end{equation}

\subsection{Fundamental limit: Kalman filter}
It is well known that if the measurements from all sensors can be collected by a single fusion center, then the centralized Kalman filter provides the optimal estimate in the sense that the trace of estimation error covariance is minimized. Therefore, the Kalman estimate acts as the fundamental limitation for all estimation schemes and will be briefly reviewed in this part. 

Let $P(k)$ be the error covariance of Kalman estimate at time $k$. Under
Assumption~\ref{assup:observable}, the error covariance will converge to the steady state exponentially fast (\cite{anderson2012optimal}):
\begin{align}
	P=\lim _{k \rightarrow \infty} P(k). \label{eqn:KFcov}
\end{align}
Since a sensor network typically operates for a long period of time, we consider the steady-state Kalman filter, which has the fixed gain
\begin{align}\label{eqn:KFgain}
	K=P C^{T}\left(C P C^{T}+R\right)^{-1}.
\end{align}
By using $K$, the optimal Kalman estimate is calculated recursively as 
\begin{equation}\label{eqnn:optimalest}
	\begin{split}
		\hat{x}(k+1) =(A-KCA)\hat{x}(k)+Ky(k+1).
	\end{split}
\end{equation}

\subsection{Framework of the proposed distributed estimator}
Clearly, Kalman filter is a centralized solution since the optimal estimate \eqref{eqnn:optimalest} fuses the measurements of all sensors. To be suitable in distributed settings, this paper aims to propose a distributed implementation of the Kalman filter such that each sensor can obtain a stable local estimate by communicating with only immediate neighbors. 

Specifically, our distributed estimation algorithm is developed based on a lossless decomposition of the centralized Kalman filter (see Fig.~\ref{fig:infoflow}), and we show that the performance of Kalman filter is equivalent a bank of local filters fused by a weighted sum. In our approach, the distributed estimator is designed as illustrated in Fig.~\ref{fig:blkdiag}. It has two phases, where the first phase implements the local filters solely based on the own measurement of each sensor and the second phase fuses the neighboring states by replacing the weighted sum in Fig.~\ref{fig:infoflow} with a synchronization procedure. In the rest of this paper, we shall detail the framework by respectively introducing the phases of decomposing the Kalman filter and synchronizing the local states.

\begin{figure}
	\centering
	\includegraphics[width=0.49\textwidth]{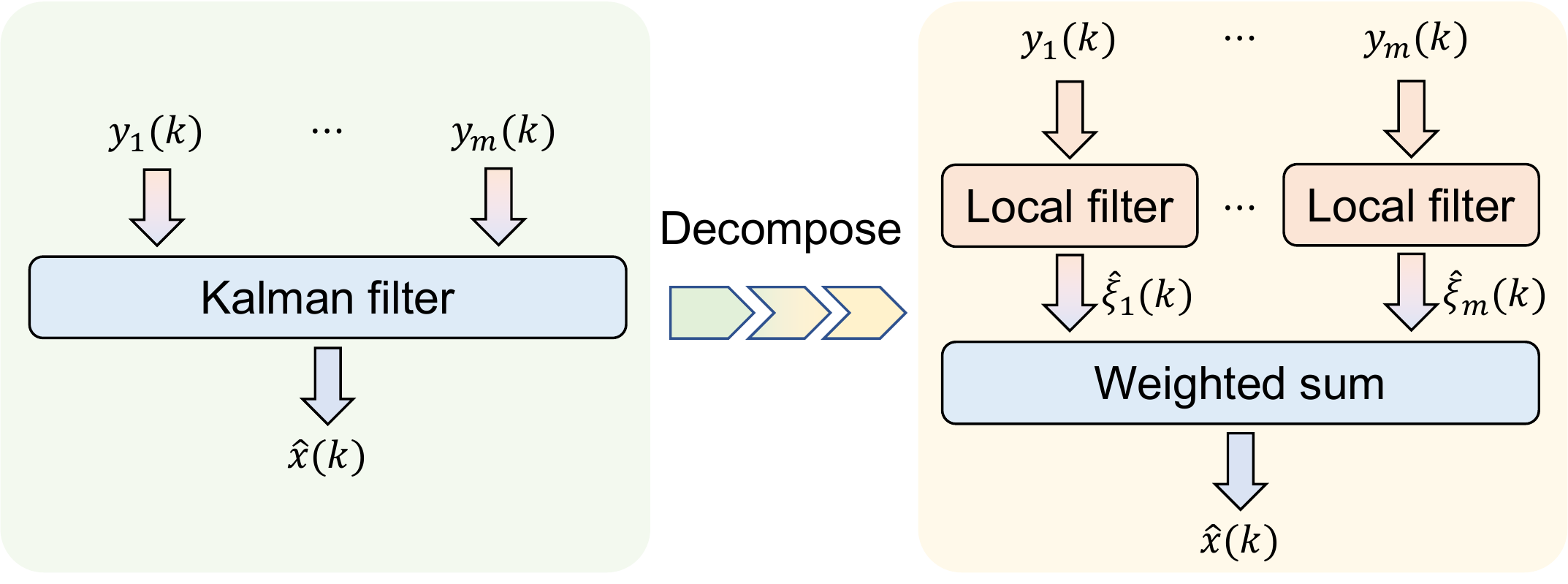}
	\caption{The information flows of centralized Kalman filter (left) and decomposition of Kalman filter \eqref{eqn:localdecompose} (right).}
	\label{fig:infoflow}
\end{figure}


\begin{figure}
	\centering
	\includegraphics[width=0.42\textwidth]{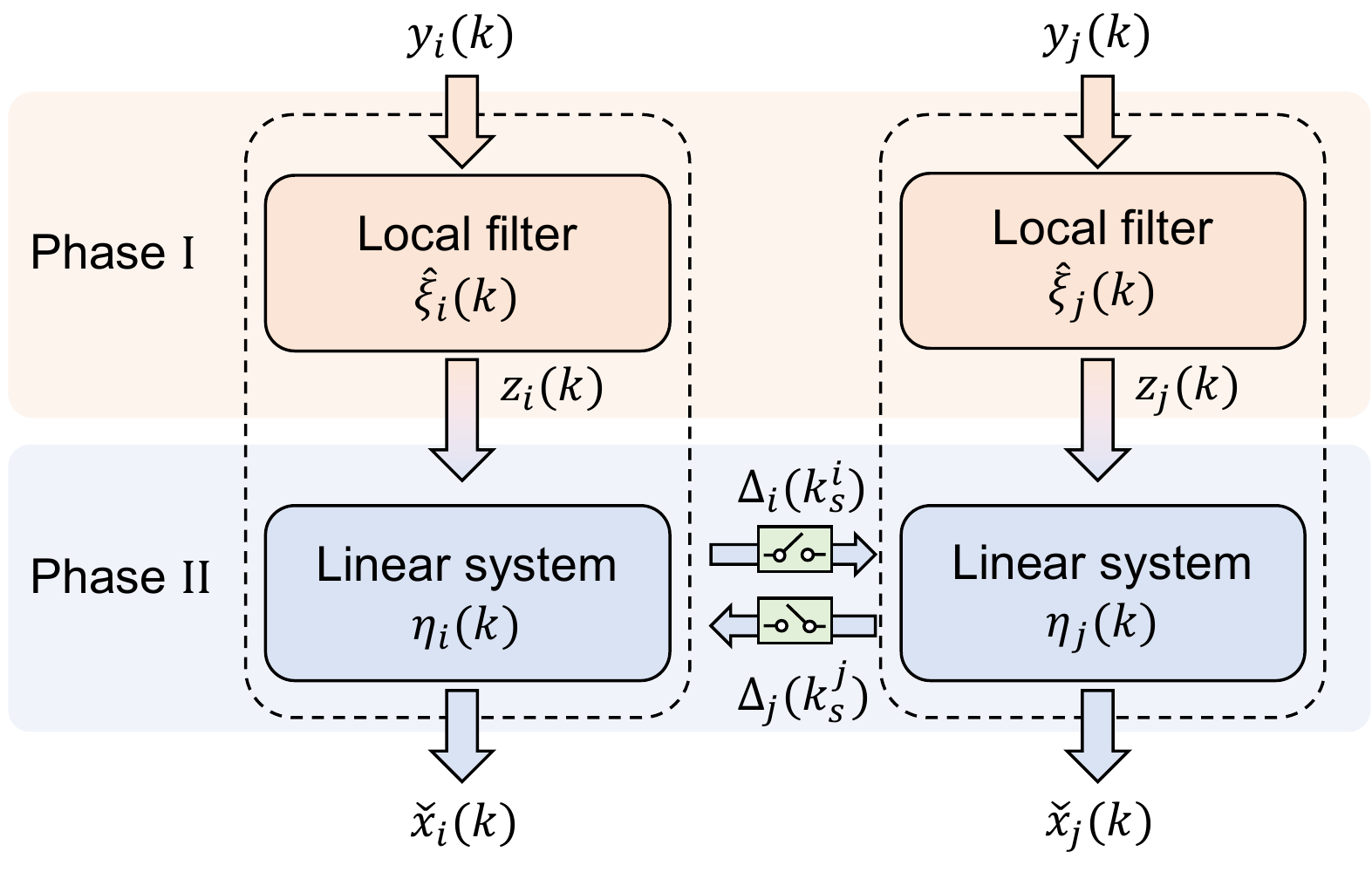}
	\caption{The information flow of the proposed distributed estimation algorithm, where nodes $i$ and $j$ are immediate neighbors.}
	\label{fig:blkdiag}
\end{figure}

\section{A Lossless Decomposition of Kalman Filter}\label{sec:decompose}
This section is devoted to a decomposition of Kalman filter. In particular, the optimal estimate \eqref{eqnn:optimalest} is proved to be a linear combination of a bank of local filters. This result is essential for us to design a framework of distributed estimation later in this paper. Moreover, this section extends in a non-trivial way, the results in \cite{yan2021distributed} by performing model reduction, which results in the decomposition of Kalman filter with lower order. Consequently, the developed distributed estimator enjoys lower message complexity, as will be discussed in Section~\ref{sec:implement}.

To begin with, without loss of any generality, suppose that the system matrix $A$ takes a decomposed form as
\begin{equation}\label{eqn:diagA}
A = \begin{bmatrix}
A^u & \\
& A^s
\end{bmatrix},
\end{equation}
where $A^u\in\mathbb{R}^{n^u\times n^u}$ and $A^s\in\mathbb{R}^{n^s\times n^s}$; any eigenvalue of $A^u$ lies on or outside the unit circle while all the eigenvalues of $A^s$ are strictly inside. We then introduce the following lemmas:
\begin{lemma}[\hspace{1pt}\cite{antsaklis1997linear}]\label{prop:controllable}
	For any $\Lambda\in\mathbb{R}^n$, if $\Lambda$ is non-derogatory\footnote{A matrix is said to be non-derogatory if every eigenvalue of it has geometric multiplicity $1$ \cite{hildebrand1987introduction}.} and in the Jordan form, then $(\Lambda,\,\1_n)$ is controllable.
\end{lemma}

\begin{lemma}[\hspace{1pt}\cite{yan2021distributed}]\label{lmm:observable}
Suppose that $(X,p)$ is controllable, where $X\in \mathbb R^{n\times n}$ and $p \in \mathbb R^n$. For any $q\in\mathbb R^n$, if $X+pq^T$ and $X$ do not share any eigenvalues, then $(X^T+qp^T,q)$ is controllable.
\end{lemma}

\begin{lemma}[\hspace{1pt}\cite{yan2021distributed}]\label{lmm:regular}
Suppose that $(X,p)$ is controllable, where $X\in \mathbb R^{n\times n}$ and $p \in \mathbb R^n$. Denote the characteristic polynomial of $X$ by $\varphi(s) = \det(sI-X)$. 
	Let $Y\in \mathbb R^{m\times m}$ and $q \in \mathbb R^m$ such that
	$
	\varphi(Y)q = 0.
	$ holds.
	Then there exists $T\in\mathbb{R}^{m\times n}$ which solves the following equations:
	\begin{equation}\label{eqn:T1}
	TX=YT,\;Tp = q.
	\end{equation}
\end{lemma}

For simplicity, we represent the Kalman gain as 
$$K = [K_1,\cdots,K_m],$$
namely, $K_i\in\mathbb{R}^n$ is the $i$-th column of $K$. Accordingly, let us rewrite the Kalman estimate \eqref{eqnn:optimalest} as
\begin{equation}\label{eqn:optimalKF}
\hat{x}(k+1) = (A-KCA)\hat{x}(k)+\sum_{i=1}^m K_iy_i(k+1).
\end{equation}

Since $(A,C)$ is observable, it is not difficult to conclude that the matrix $A-KCA$ is strictly stable. Then based on it, one can always construct a Jordan matrix $\Lambda\in \mathbb R^{n\times n}$ satisfying that
\begin{enumerate}
\item $\Lambda$ is strictly stable and non-derogatory.
\item The characteristic polynomials of $\Lambda$ and $A-KCA$ are same. 
\end{enumerate} 
By virtue of Lemma~\ref{prop:controllable}, one knows that $(\Lambda, \1_n)$ is controllable. Then, as guaranteed by Lemma~\ref{lmm:regular}, for each $i=1,\cdots,m$, a matrix $F_i\in\mathbb{R}^{n\times n}$ exists which solves the following equations:
\begin{equation}\label{eqn:F_i}
F_i \Lambda = (A-KCA)F_i,\;F_i \mathbf 1_n = K_i.
\end{equation}

\begin{remark}
 The solution of \eqref{eqn:F_i}, i.e., $F_i$, can be obtained by following the construction proof of Lemma~\ref{lmm:regular}, which is provided in \cite{yan2021distributed}.
\end{remark}
Based on the property that $(\Lambda, \1_n)$ is controllable, we would also design $\beta\in \mathbb{R}^n$ and compute
\begin{equation}\label{eqn:Sdef}
S  = \Lambda + \1_n\beta^T
\end{equation}
such that the pair $(S^T, \beta)$ is controllable.
To achieve this, we should use Lemma~\ref{lmm:observable} and design $\beta$ and $S$ to satisfy the following statements:
\begin{enumerate}
	\item Let $\phi(s)$ and $\psi(s)$ be the characteristic polynomials of $S$ and $A^u$, respectively. Then it should hold that $\psi(s)|\phi(s)$. Namely, there exists a polynomial $\sigma(s)$ such that
	\begin{equation}
		\phi(s) = \psi(s) \sigma(s).
	\end{equation} 
	\item Any root of $\sigma(s)$ is strictly within the unit circle but not an eigenvalue of $\Lambda$. 
\end{enumerate}
Therefore, the unstable and stable eigenvalues of $S$ are the roots of $\psi(s)$ and $\sigma(s)$, respectively. Moreover, the unstable ones should be coincide with the eigenvalues of $A^u$, while the stable ones can be freely designed (but should not be the eigenvalues of $\Lambda$). Since $\Lambda$ is strictly stable, we conclude that $S$ does not share any eigenvalues with $\Lambda$. As a result of Lemma~\ref{lmm:observable}, $(S^T,\beta)$ is controllable.

\begin{remark}\label{rmk:betadef}
One can find $\beta$ and $S$ by following the procedure below:
\begin{enumerate}
	\item[i)] Pre-determine the eigenvalues of $S$ such that the unstable eigenvalues are identical to the ones of $A^u$, while the stable ones are not the same with those of $\Lambda$.
	\item[ii)] Through pole placement, find $\beta$ such that the eigenvalues of $\Lambda + \1_n\beta^T$ are placed at the desired locations. Notice that $\beta$ always exists since $(\Lambda,\textbf{1}_n)$ is controllable. 
	\item[iii)] Calculate $S$ by \eqref{eqn:Sdef}.
\end{enumerate}
\end{remark}

Now we are ready to design the local filters, which are performed by each sensor $i$ solely based on its own measurement $y_i(k)$. Notice that the system matrix $A$ may have unstable modes, implying $y_i(k)$ is not necessarily stable\footnote{In this paper, we say a signal is \textit{stable} if the covariance of it is bounded at all time.}, namely, the covariance of it can be unbounded. Therefore, instead of $y_i(k)$, we would like to design the local filter by using a stable signal $z_i(k)$ given as below:
\begin{equation}\label{eqn:xi_z}
\begin{split}
z_i(k) &= y_i(k+1)-\beta^T\hat{\xi}_i(k),\\
\hat{\xi}_i(k+1) &= S \hat{\xi}_i(k)+\1_n z_i(k),
\end{split}
\end{equation}
where $\hat\xi_i(k)$ is the output of the local filter with $\hat{\xi}_i(0) = 0$, and $\beta$ and $S$ are obtained by Remark~\eqref{rmk:betadef}. The following lemma shows that the optimal Kalman filter is indeed a linear combination of the local filters $\hat \xi_i(k), i= 1,\cdots,m$, and thus can be losslessly recovered by them. Moreover, for any sensor $i$, the signal $z_i(k)$ is stable.
For the sake of legibility, the proof is provided in Appendix~\ref{sec:app_localeqv}.
\begin{lemma}\label{lmm:localeqv}
	Suppose that each sensor implements the local filter \eqref{eqn:xi_z}. Then the following statements hold at any $k$: 
	\begin{enumerate}
		\item For any sensor $i$, the covariance of $z_i(k)$ is bounded.
		\item The optimal Kalman estimate \eqref{eqnn:optimalest} can be losslessly recovered from the local estimates $\hat \xi_i(k), i=1,2,\cdots,m,$ by
		\begin{equation}\label{eqn:localdecompose}
		\hat{x}(k) = \sum_{i=1}^{m}F_i \hat\xi_i(k),
		\end{equation}
		where $F_i$ is the solution of \eqref{eqn:F_i}. 
	\end{enumerate}
\end{lemma}

The information flows of the Kalman filter \eqref{eqnn:optimalest} and the decomposition of it \eqref{eqn:localdecompose} are depicted in Fig.~\ref{fig:infoflow}. Notice that the decomposition here is based on the approach of \cite{mo2016secure}.
There, the local filters are to be used in a centralized manner, allowing a simpler structure
with $\beta=\0_{n}$. The local filters of the form \eqref{eqn:xi_z} was proposed in \cite{yan2021distributed} for distributed estimation. Lemma~\ref{lmm:localeqv} is included in the results there, but not explicitly. Here we present the proof for the sake of completeness.

\begin{remark}
We further remark that Lemma~\ref{lmm:localeqv}, by following a similar procedure, can be extended to any stable Luenberger estimators which may not necessarily be the Kalman filter. Also, the results can be generalized to consider other noise models. For example, we can decompose the $H_\infty$ estimator in the presence of bounded noise. 
\end{remark}


\begin{remark}\label{rmk:z}
We also note that it is essential to design the local filter \eqref{eqn:xi_z} with stable input $z_i(k)$, since the boundedness of its covariance is key to guarantee the synchronization among local states. This will further help to establish the stability of local estimators as we will see later in Theorem~\ref{thm:synchronization}.
\end{remark}

We next present \eqref{eqn:localdecompose} in a matrix form. To this end, let us denote by $r$ the rank of Kalman gain $K\in\mathbb{R}^{n\times m}$, namely,
\begin{equation}\label{eqn:rdef}
r \triangleq \rank(K) \leq \min\{m,n\}.
\end{equation}
As a result, there exists a matrix $V\in \mathbb R^{r\times m}$ which is of rank $r$ such that $K$ can be decomposed as
\begin{equation}\label{eqn:Kdecompose}
\begin{aligned}
K&=\begin{bmatrix}
K_1&\cdots&K_m	
\end{bmatrix}= \begin{bmatrix}
\tilde K_1&\cdots&\tilde K_r	
\end{bmatrix} V=\tilde KV,
\end{aligned}
\end{equation}
where $\tilde K_i, \;i=1,\cdots,r,$ are linearly independent. Then by \eqref{eqn:localdecompose}, let us rewrite the Kalman estimate as below:
\begin{equation}\label{eqn:kf_matrix}
\begin{aligned}
&\;\quad\hat{x}(k+1)=\sum_{i=1}^m F_i \hat \xi_i(k+1) \\&= \sum_{i=1}^m F_i (\Lambda+\mathbf 1\beta^T)\hat \xi_i(k) + \sum_{i=1}^m F_i \1_n z_i(k) \\
&= (A-KCA)\sum_{i=1}^m F_i \hat \xi_i(k) + \sum_{i=1}^m K_i \beta^T \hat \xi_i(k) +\sum_{i=1}^m K_i z_i(k),
\end{aligned}
\end{equation}
where the second and third equalities hold respectively by \eqref{eqn:xi_z} and \eqref{eqn:F_i}. We next consider the second term of RHS, i.e., $\sum_{i=1}^m K_i \beta^T \hat \xi_i(k)$. As a result of \eqref{eqn:Kdecompose}, it follows that 
\begin{equation}\label{eqn:titleK}
\sum_{i=1}^m K_i \beta^T \hat \xi_i(k)= \sum_{i=1}^r\tilde K_i \beta^T \sum_{j=1}^m v_{ij}\hat \xi_j(k),
\end{equation}
where $v_{ij}$ is the $(i,j)$th entry of $V$. Moreover, it follows from \eqref{eqn:xi_z} that
\begin{equation}\label{eqn:vij}
\sum_{j=1}^m v_{ij}\hat \xi_j(k+1) = S\sum_{j=1}^m v_{ij}\hat \xi_j(k)+\sum_{j=1}^m v_{ij}\1_nz_i(k).
\end{equation}

To simplify notations, let us denote 
\begin{equation}\label{eqn:vartheta}
\vartheta(k)\triangleq \begin{bmatrix}
\sum_{i=1}^m F_i \hat \xi_i(k)\\
\sum_{j=1}^mv_{1j}\hat \xi_j(k)\\
\vdots\\
\sum_{j=1}^mv_{rj}\hat \xi_j(k)\\
\end{bmatrix}\in\mathbb{R}^{n(r+1)}.
\end{equation}
It thus follows from \eqref{eqn:kf_matrix}--\eqref{eqn:vij} that
\begin{equation}\label{eqn:theta}
\vartheta(k+1)=H\vartheta(k)+
Lz(k),
\end{equation}	
where\begin{equation}\label{eqn:S}
\begin{split}
&H\triangleq \begin{bmatrix}
A-KCA &\tilde K_1\beta^T&\cdots &\tilde K_r \beta^T\\	
&S&&\\
&&\ddots&\\
&&&S
\end{bmatrix}\in\mathbb{R}^{n(r+1)\times n(r+1)}, \\
&L \triangleq  \begin{bmatrix}
K\\
V\otimes \1_n
\end{bmatrix}\in\mathbb{R}^{n(r+1)\times m},\\
&z(k)\triangleq \begin{bmatrix}
z_1(k) & \cdots &z_m(k)
\end{bmatrix}^T\in \mathbb{R}^m.
\end{split}
\end{equation}
For convenience, we denote by $L_i$ the $i$-th column of $L$, namely,
\begin{equation}\label{eqn:Li}
L=	\begin{bmatrix}
	L_1&\cdots&L_m	
	\end{bmatrix}.
\end{equation}
In view of \eqref{eqn:kf_matrix}, the optimal Kalman estimate $\hat{x}(k)$ is indeed the vector consisting of the first $n$ entries of $\vartheta(k)$. Therefore, the optimal estimate can be losslessly recovered by \eqref{eqn:xi_z} and \eqref{eqn:theta}, where a center is however required to fuse $\hat{\xi}_i(k)$ and $z_i(k)$ from all sensors. In the rest of this paper, we will show how to use \eqref{eqn:xi_z} and \eqref{eqn:theta} to design a distributed implementation of the Kalman filter.

\begin{remark}
Our decomposition approach here is an extension of \cite{yan2021distributed} where the Kalman gain $K$ is used directly. However, there, to achieve a distributed implementation of it, the message complexity should be $\min\{m,n\}$. In contrast, \eqref{eqn:theta} can be implemented with lower message complexity $r$. Moreover, notice that $r$ is defined as the rank of estimation gain $K$. In general, however, the rank of estimation gain $K$ may not have a lower rank. Later in Section~\ref{sec:Kdesign} of the paper, we will provide a design method to find an estimation gain matrix for a given rank while minimizing the estimation error. This will allow us to reduce the message complexity with some tradeoff in its estimation performance.
\end{remark}

\section{An Event-Based Distributed Implementation of Kalman Filter}\label{sec:implement}
This section is devoted to a distributed implementation of the Kalman filter with reduced communication among sensors. The traditional approaches \cite{olfati2007distributed,olfati2005distributed,olfati2009kalman,battistelli2014consensus,li2011consensus,battistelli2016stability,del2009distributed,kar2010gossip,ma2016gossip,cattivelli2010diffusion,hu2011diffusion,cattivelli2008diffusion,farina2010distributed,haber2013moving} require that each sensor node broadcasts its local information to neighbors at least once during the sampling interval. This inevitably causes a large number of data transmission, which leads to an increased communication burden and a shortened lifetime of the sensor network. From this perspective, this section also presents event-based communication strategies to reduce the transmission frequencies for each sensor node.

\subsection{Framework of the event-based distributed estimation}
We shall leverage the results established in Section~\ref{sec:decompose} to design a distributed estimator with event-based communication strategies. Specifically, based on the decomposition of Kalman filter, for any sensor $i$, the update of it during each sampling period can be divided into two phases. In Phase I, sensor $i$ performs the local filter \eqref{eqn:xi_z} solely using its own measurement without communicating with others. On the other hand, Phase II fuses the neighboring information based on \eqref{eqn:theta}. 

In view of \eqref{eqn:theta}, it is clear that the Kalman estimate fuses $\hat{\xi}_i(k)$ and $z_i(k)$ from all sensors. However, since each local sensor is only capable of accessing the information in its neighborhood, we aim to implement Kalman filter in a distributed fashion by running a synchronization algorithm. For the particular purpose of decreasing the transmission frequency, event-based communication strategies will be adopted.

To be concrete, let each sensor $i$ keep a local state as below:
\begin{equation}\label{eqn:etadef}
\eta_i(k) \triangleq
{\left[\begin{array}{c}
	\eta_{0,i}(k) \\
	\eta_{1, i}(k) \\
	\vdots \\
	\eta_{r, i}(k)
	\end{array}\right]\in\mathbb{R}^{n(r+1)},} 
\end{equation}
where $\eta_{j, i}(k) \in \mathbb{R}^n, j=0,1,\cdots,r$. In order to approach the performance of Kalman filter, the local state will be updated through the following synchronization algorithm:
\begin{equation}\label{eqn:update}
\eta_i(k+1) =H\eta_i(k) +L_iz_i(k)+B\sum_{j=1}^m a_{ij}(\widehat{\Delta}_j(k)-\widehat{\Delta}_i(k)),
\end{equation}
where $\eta_i(0)=0$, $H$ and $L_i$ are respectively defined in \eqref{eqn:S} and \eqref{eqn:Li}, and
\begin{equation}
B = 
\begin{bmatrix}
\0_{n\times r}\\
I_r\otimes \1_n
\end{bmatrix}\in\mathbb{R}^{n(r+1)\times r}.
\end{equation}
Moreover, $\widehat \Delta_{i}(k)\in\mathbb{R}^r$ is the latest information broadcast by sensor $i$, and is calculated by
\begin{equation}
\begin{split}
&\widehat \Delta_{i}(k) = T \widehat \eta_i(k),
\end{split}
\end{equation}
where 
\begin{equation}
\widehat \eta_{i}(k) = H^{(k-k_{s}^{i})}\eta_i(k_{s}^{i}),\; \; k\in [k_{s}^{i}, k_{s+1}^{i}),
\end{equation}
\begin{equation}\label{eqn:T}
T = 
\begin{bmatrix}
\0_{r\times n}& I_r\otimes \Gamma
\end{bmatrix}\in\mathbb{R}^{r\times n(r+1)},
\end{equation}
and $\Gamma\in\mathbb{R}^{1\times n}$ is the parameter to be designed. 

To determine the triggering instants $k_{s}^{i}$, each sensor $i$ considers a triggering function $f_i(k)$ in the following form:
\begin{equation}\label{eqn:triggerfun}
	f_i(k)= ||\epsilon_{i}(k)||^2-h_i(k),
\end{equation}
where
\begin{equation}\label{eqn:epsilon}
	\epsilon_{i}(k)= \widehat{\eta}_i(k)-\eta_i(k), 
\end{equation} 
and $h_i(k)$ is a threshold function as will be discussed later in Section~\ref{sec:trigger}. The sensor updates its local state $\eta_i(k)$ based on \eqref{eqn:update} until the triggering function \eqref{eqn:triggerfun} exceeds $0$. Particularly, once $f_i(k) \geq 0$, agent $i$ will be triggered. It then broadcasts $\Delta_{i}(k)$ to neighbors, resetting $\epsilon_{i}(k)$ to zero. Obviously, the sequence of triggering instants is determined recursively as
\begin{equation}\label{eqn:triggertime}
	k_{s+1}^{i} \triangleq \min \left\{k>k_{s}^{i} \mid f_i(k) \geq 0\right\},\; k_{0}^{i}=0.
\end{equation}
 
\begin{remark}\label{rmk:rankK}
	Notice that, instead of directly transmitting the local state $\widehat \eta_{i}(k)\in\mathbb{R}^{n(r+1)}$, each sensor node broadcasts a ``coded" vector $\widehat \Delta_i(k)\in\mathbb{R}^r$. Therefore, the data size for each transmission is $r=\rank(K)\leq \min\{m,n\}$. As compared with existing works, e.g., \cite{olfati2005distributed,olfati2009kalman,battistelli2014consensus,li2011consensus,battistelli2016stability}, which usually require information exchange on the local covariance matrix of size $n^2$, the proposed algorithm enjoys lower message complexity.
\end{remark} 

By collecting Phases I and II together, the update of any sensor $i$ is summarized in Algorithm~\ref{alg:dist_est}. Fig.~\ref{fig:blkdiag} presents the information flow of Algorithm~\ref{alg:dist_est}, which requires no fusion center and is achieved in a distributed manner. As compared with Fig.~\ref{fig:existing}, the novelty of the proposed algorithm lies in the decoupling of the local filter from the fusion process. Therefore, the communication occurs only in Phase II, and the performance of local filters will not be affected even when no sensors are triggered to transmit at certain times. 

\begin{algorithm}
	1:\: (Phase I) Solely using its own measurement, sensor $i$ computes $z_i(k)$ and updates the state of the local filter by \eqref{eqn:xi_z}.\\
	2:\: (Phase II) By fusing the information most recently received from its neighborhood, sensor $i$ updates $\eta_i(k+1)$ according to the synchronization algorithm \eqref{eqn:update}--\eqref{eqn:T}.\\
	3:\: Sensor $i$ obtains the local estimate as
	\begin{equation}\label{eqn:localfuse}
		\breve{x}_i(k+1)=m\eta_{0,i}(k+1).
	\end{equation}
	4:\: Sensor $i$ checks the triggering function \eqref{eqn:triggerfun}. Once $f_i(k)\geq 0$, it
	broadcasts $\Delta_i(k+1)$ to neighbors. 
	\caption{An event-based distributed estimation algorithm for sensor $i$ at time $k>0$}
	\label{alg:dist_est}
\end{algorithm}

\subsection{A general class of triggering functions}\label{sec:trigger}
An important feature of the triggering function \eqref{eqn:triggerfun} is that it ensures $||\epsilon_{i}(k)||^2$ to be smaller than the threshold $h_i(k)$. This happens because once sensor $i$ finds $f_i(k) \geq 0$, the event is triggered, which resets $\epsilon_{i}(k)=0$. Instead of proposing any specific triggering function, we show that a general class of triggering strategies can be applied in our framework for yielding stable distributed estimates. Specifically, we require that the threshold function $h_i(k)$ is designed such that $||\epsilon_{i}(k)||^2$ is upper bounded by some $\hbar<\infty$, namely,  
\begin{equation}\label{eqn:hbar}
||\epsilon_{i}(k)||^2\leq \hbar, \;\forall k\geq 0.
\end{equation}

We now present several triggering functions that are commonly used in the literature and also detail what $h_i(k)$ is in each case. It is straightforward to show that \eqref{eqn:hbar} is guaranteed.
\begin{enumerate}
	\item Static time-dependent triggering function (\hspace{1pt}\cite{kadowaki2014event}):
	\begin{equation}\label{eqn:trigger1}
		\begin{split}
			h_{i}(k)&=c_0+c_1\alpha^{k},\\
		\end{split}
	\end{equation}
	where $c_0>0$, $c_1\geq 0$, and $\alpha\in(0,1)$.
	\item Static state-dependent triggering function (\hspace{1pt}\cite{yi2018dynamic,mishra2021dynamic}):
	\begin{equation}\label{eqn:trigger2}
		\begin{split}
			\widehat{q}_{i}(k)&=\min \left\{\frac{1}{2} \sum_{j=1}^m a_{ij}\big|\big|\widehat \Delta_j(k_{s}^{j})-\widehat \Delta_i(k_{s}^{i})\big|\big|^{2}, \ell\right\},\\
			h_{i}(k)&=\alpha_{i}(k) \widehat{q}_{i}(k),
		\end{split}
	\end{equation}
where $\ell>0$ and $\alpha_{i}(k)$ takes nonnegative values and exponentially decreases to zero.
	\item Dynamic triggering function (\hspace{1pt}\cite{yi2018dynamic,mishra2021dynamic}):
	\begin{equation}
	\begin{split}
		\chi_{i}(k+1)&=\beta_{i} \chi_{i}(k)+\alpha_{i}(k) \widehat{q}_{i}(k)-||\epsilon_{i}(k)||^{2},\\
h_{i}(k)&= \frac{1}{\theta_{i}}\chi_{i}(k)+\alpha_{i}(k) \widehat{q}_{i}(k),
	\end{split}
	\end{equation}
where $\chi_{i}(0)>0$, $\beta_i\in(0,1)$ and $\theta_{i}>1/\beta_i$. Moreover, $\widehat{q}_{i}(k)$ and $\alpha_{i}(k)$ are defined in \eqref{eqn:trigger2}.
\end{enumerate}

One merit of our framework is that by decoupling the local filters from the fusion process, we can reformulate the problem of distributed estimation to that of stochastic linear systems synchronization. In the next section, we will prove that any event-based algorithm guaranteeing \eqref{eqn:hbar} can facilitate the synchronization of stochastic linear systems, and thus contribute to establish stable distributed estimators.
However, as one might imagine, different triggering functions result in different triggering frequencies and estimation accuracy.
 
\section{Estimation Performance Analysis}\label{sec:analysis}
This section will theoretically analyze the performance of Algorithm~\ref{alg:dist_est}. We remark that in our previous work \cite{yan2021distributed}, we have provided a unified framework for studying the performance of distributed estimators in the scenarios where communication among agents is independent of the system states and sensor measurements. However, in Algorithm~\ref{alg:dist_est}, the communication inevitably relies on these states as it is triggered by certain events depending on them. This prevents the methodologies in \cite{yan2021distributed} from being directly applied. Therefore, in this paper, we would first resort to $c$-martingale convergence lemma (as proposed in Appendix~\ref{sec:app_synchronization}) establishing the mean-squared synchronization of local states of all sensors, namely $\eta_{i}(k)$'s. This result will next be leveraged to prove the stability of the distributed estimators.

\subsection{Synchronization of local states}\label{sec:problem}

In order to show the synchronization among local states, let us introduce the following lemma:

\begin{lemma}\label{lmm:eta}
	Suppose that the Mahler measure\footnote{The Mahler measure of a matrix is defined as the absolute product of its unstable eigenvalues.} of matrix $S$ meets the following condition:
	\begin{equation}\label{eqn:unstable}
		\prod_j |\lambda_j^u(S)| < \frac{1+\mu_2/\mu_m}{1-\mu_2/\mu_m},
	\end{equation}
	where $\lambda_j^u(S)$ represent the unstable eigenvalues of $S$, and $\mu_2$ and $\mu_m$ are defined in \eqref{eqn:mu}. Let
\begin{equation}\label{eqn:Gamma}
	\Gamma=\frac{2}{\mu_2+\mu_m}\frac{\1_n^T\mathcal{P}S}{\1_n^T\mathcal{P}\1_n}\in\mathbb{R}^{1\times n},
\end{equation}
where $\mathcal{P}>0$ solves the following modified algebraic Riccati inequality:
\begin{equation}\label{eqn:riccati}
	\mathcal{P}-S^{T} \mathcal{P} S+\left(1-\zeta^{2}\right) \frac{S^{T} \mathcal{P} \1_n \1_n^{T} \mathcal{P} S}{\1_n^{T} \mathcal{P} \1_n}>0,
\end{equation}
and $\zeta$ satisfies that 
\begin{equation}\label{eqn:zeta}
\prod_{j}\left|\lambda_{j}^{u}(S)\right|<\zeta^{-1} \leq\frac{1+\mu_{2} / \mu_{m}}{1-\mu_{2} / \mu_{m}}.
\end{equation}
Then for any $j\in\{2,...,n\}$, it holds that 
\begin{equation}\label{eqn:Seig}
	\rho(H-\mu_j BT)<1.
\end{equation}
\end{lemma}
\begin{proof}
Consider any $j\in\{2,\...,n\}$. It follows that
\begin{equation}\label{eqn:Hi}
	\begin{split}
	H-\mu_j BT =\begin{bmatrix}
		A-KCA &\tilde K_1\beta^T&\cdots &\tilde K_r \beta^T\\	
		&S-\mu_j\1_n\Gamma&&\\
		&&\ddots&\\
		&&&S-\mu_j\1_n\Gamma
	\end{bmatrix}.
	\end{split}
\end{equation}
Recall that $(\Lambda, \1_n)$ is controllable. In view of \eqref{eqn:S}, it is not difficult to verify that $(S, \1_n)$ is also controllable.  Hence, by the choice of $\zeta$, there exists $\mathcal{P}>0$ that solves \eqref{eqn:riccati}. Together with \eqref{eqn:Gamma}, it holds that
\begin{equation}
	\begin{split}
		&\quad(S-\mu_j \1_n \Gamma)^T \mathcal{P} (S-\mu_j \1_n \Gamma)-\mathcal{P}\\
		&= S^T\mathcal{P}S-(1-\zeta_j^2)\frac{S^{T} \mathcal{P} \1_n \1_n^{T} \mathcal{P} S}{\1_n^{T} \mathcal{P} \1_n}-\mathcal{P}\\
		&\leq  S^T\mathcal{P}S-(1-\zeta^2)\frac{S^{T} \mathcal{P} \1_n\1_n^{T} \mathcal{P} S}{\1_n^{T} \mathcal{P} \1_n}-\mathcal{P}<0,
	\end{split}
\end{equation}
where $\zeta_j = 1-2\mu_j/(\mu_2+\mu_m)$, and the first inequality holds by \eqref{eqn:zeta}. Therefore, the Lyapunov inequality holds with $\mathcal{P}$, and one concludes that $\rho(S-\mu_j \1_n\Gamma)<1$. Notice that $A-KCA$ is stable. Our proof is thus completed. 
\end{proof}

The estimation performance of the proposed framework is expressed with respect to the average of local states of all sensors given by
\begin{align}
	\bar{\eta}(k) \triangleq \frac{1}{m}\sum_{i=1}^m \eta_i(k).\label{eqn:bareta}
\end{align}
The synchronization among local states is formally stated as follows:
\begin{theorem}\label{thm:synchronization}
Suppose that the condition \eqref{eqn:unstable} holds, and $\Gamma$ is designed based on \eqref{eqn:Gamma} and \eqref{eqn:riccati}. By applying the synchronization algorithm \eqref{eqn:update} with an event-based communication strategy that guarantees \eqref{eqn:hbar}, synchronization among local states is reached in the mean square sense. That is, the following statements hold at any time $k$:
	\begin{enumerate}
		\item \textit{Consistency condition:}
		\begin{equation}\label{eqn:consistency2}
			\bar{\eta}(k+1) =H\bar{\eta}(k)+ \bar L_z(k),
		\end{equation}
	where $\bar L_z(k) \triangleq\frac{1}{m}\sum_{i=1}^m L_iz_i(k)$.
		\item \textit{Consensus condition:} There exists $\Xi>0$ such that 
		\begin{equation}\label{eqn:consensus2}
			\cov[\eta_i(k)-\bar{\eta}(k)]\leq \Xi, \;\forall k.
		\end{equation}
	\end{enumerate} 
\end{theorem}
\begin{proof}
The proof is provided in Appendix~\ref{sec:app_synchronization}.
\end{proof}

The consistency condition \eqref{eqn:consistency2} claims that the dynamics of the average state $\bar{\eta}(k)$ is governed by $z(k)$ only. Therefore, the interaction among sensors only affects the evolution of local states but not $\bar{\eta}(k)$. On the other hand, \eqref{eqn:consensus2} states that, despite the signal $z(k)$, each local state can track $\bar{\eta}(k)$ with bounded error covariance. These conditions would help to establish the stability of local estimators.

\subsection{Stability analysis of local estimators}\label{sec:stable}
In Theorem~\ref{thm:synchronization}, we have proven that the synchronization algorithm \eqref{eqn:update} guarantees that the local states achieve both the consistency and consensus conditions. We shall, in this subsection, show how these conditions will help to achieve a stable local estimate at each sensor side.

First, the next theorem shows that the average of local estimates from all sensors is indeed the optimal Kalman estimate  \eqref{eqnn:optimalest}, as guaranteed by the consistency condition \eqref{eqn:consistency2}:

\begin{theorem}\label{thm:optimal}
Suppose that the condition \eqref{eqn:unstable} holds, and $\Gamma$ is designed based on \eqref{eqn:Gamma} and \eqref{eqn:riccati}. By performing Algorithm~\ref{alg:dist_est}, it holds at any $k\geq 0$ that
	\begin{equation}
	\frac{1}{m}\sum_{i=1}^m \breve{x}_i(k)=\hat{x}(k).
	\end{equation}
Namely,	the average of local estimates from all sensors coincides with the Kalman estimate.
\end{theorem}
\begin{proof}
As a result of consistency condition \eqref{eqn:consistency2}, it follows for any $j=1,\cdots,r$ that
\begin{equation}\label{eqn:etasum}
	\sum_{i=1}^m \eta_{j,i}(k+1) =S\sum_{i=1}^m\eta_{j,i}(k) +\sum_{i=1}^mv_{ji}\1_n z_i(k). 
\end{equation}
Comparing it with \eqref{eqn:xi_z}, we can obtain the following statement for any time $k$ and any $j\in\mathcal{V}$:
	\begin{equation}
	\sum_{i=1}^mv_{ji}\hat \xi_i(k) = \sum_{i=1}^m \eta_{j,i}(k).
	\end{equation}
	Therefore, the following relation holds at any $k\geq 0$:
	\begin{equation}
		\begin{split}
			&\sum_{i=1}^m \eta_{0,i}(k+1)= (A-KCA)\sum_{i=1}^m \eta_{0,i}(k)\\&\quad\qquad\qquad\qquad+\sum_{i=1}^{m}\sum_{j=1}^r\tilde{K}_j\beta^T\eta_{j,i}(k)+\sum_{i=1}^{m}K_iz_i(k)
		\\&= (A-KCA)\sum_{i=1}^m \eta_{0,i}(k)+\sum_{j=1}^r\tilde{K}_j\beta^T\sum_{i=1}^{m}\eta_{j,i}(k)\\&\quad+\sum_{i=1}^{m}K_iz_i(k)
			\\&= (A-KCA)\sum_{i=1}^m \eta_{0,i}(k)+\sum_{j=1}^r\tilde{K}_j\beta^T\sum_{i=1}^mv_{ji}\hat \xi_i(k)\\&\quad+\sum_{i=1}^{m}K_iz_i(k).
		\end{split}
	\end{equation}
Comparing it with \eqref{eqn:kf_matrix} and \eqref{eqn:titleK}, one concludes that 
\begin{equation}\label{eqn:xvseta}
\hat{x}(k)=\sum_{i=1}^m \eta_{0,i}(k)=\frac{1}{m}\sum_{i=1}^m \breve{x}_i(k).
\end{equation}
This completes the proof.
\end{proof}

On the other hand, we shall also analyze the stability of estimation error, i.e., the boundedness of its covariance. This is particularly established by the consensus condition \eqref{eqn:consensus2}, as stated in the following theorem:

\begin{theorem}\label{thm:stable}
Suppose that the condition \eqref{eqn:unstable} holds, and $\Gamma$ is designed based on \eqref{eqn:Gamma} and \eqref{eqn:riccati}. By performing Algorithm~\ref{alg:dist_est}, it holds at any $k\geq 0$ that
\begin{equation}
\cov( \breve{x}_i(k)-x(k))<\infty, \;\forall i.
\end{equation}
Namely, the error covariance of each local estimate is bounded.
\end{theorem}

\begin{proof}
Let us consider the local estimator of any sensor $i$. By virtue of \eqref{eqn:consensus2}, we conclude that $\cov(\eta_{0,i}(k)-\bar{\eta}_{0}(k))$ is bounded at any time $k$, 
where $\bar{\eta}_{0}(k)=\frac{1}{m}\sum_{i=1}^m \eta_{0,i}(k)$. Then in order to prove the boundedness of $\cov( \breve{x}_i(k)-x(k))$, let us denote
\begin{equation}
	\bar{e}_i(k) \triangleq \breve{x}_i(k)-\hat{x}(k),
\end{equation}
which is the distance between local estimate $\breve{x}_i(k)$ and the optimal Kalman filter. Combining it with \eqref{eqn:xvseta} yields
\begin{equation}\label{eqn:bar_e}
	\bar{e}_i(k) = m(\eta_{0,i}(k)-\bar{\eta}_{0}(k)).
\end{equation}
Therefore, the estimation error of sensor $i$ is calculated as
\begin{equation}\label{eqn:breve_e}
	\begin{split}
		\breve{e}_i(k) &= \breve{x}_i(k)-x(k)\\&=(\breve{x}_i(k)-\hat{x}(k))+(\hat{x}(k)-x(k)) \\&= \bar{e}_i(k)+\hat{e}(k),
	\end{split}
\end{equation}
where $\hat{e}(k)$ is the estimation error of Kalman filter. According to orthogonality principle \cite{kay1993fundamentals}, $\bar{e}_i(k)$ is orthogonal to $\hat{e}(k)$. Therefore, it follows that
\begin{equation}\label{eqn:error}
\begin{aligned}
\cov(\breve{e}_i(k)) &= \cov(\bar{e}_i(k))+\cov(\hat{e}(k))\\&=m^2\cov(\eta_{0,i}(k)-\bar{\eta}_{0}(k))+P,
\end{aligned}
\end{equation}
where $P$, defined in \eqref{eqn:KFcov}, is the steady-state error convariance of Kalman filter. Since $\cov(\eta_{0,i}(k)-\bar{\eta}_{0}(k))$ is bounded, we therefore complete the proof.
\end{proof}

In view of Theorems \ref{thm:optimal} and \ref{thm:stable}, we conclude that each sensor yields a stable local estimate. This further indicates that, the problem of distributed state estimation can be resolved by using the algorithms designed for realizing the synchronization among stochastic linear systems.
 Moreover, by virtue of \eqref{eqn:error}, the performance gap between the proposed estimator and the Kalman filter is purely introduced by the consensus error $\cov[\eta_{0,i}(k)-\bar{\eta}_{0}(k)]$. 
Recalling the proof of Theorem~\ref{thm:synchronization}, we conclude that this error is introduced by the stochastic signals $\{z_i(t)\}_{t\leq k}$ and the event-triggering function. Hence, one can tune the estimation accuracy by choosing different triggering mechanisms, where the trade-off with triggering frequency should also be taken into consideration.

\begin{remark}
In Theorem~\ref{thm:stable}, we have proven the stability of estimation error. However, due to the communication triggered by noisy states, it is difficult to calculate the exact bound of estimation error. This aspect is similar to the studies in the existing event-based distributed estimation algorithms \cite{liu2018event,battistelli2018distributed,yu2020event,wu2016finite,song2019event,shi2016event,liu2015event,yan2014distributed}. 
\end{remark}

\subsection{Application to synchronization of stochastic linear systems}
By performing the decomposition of Kalman filter, we have shown that the algorithms designed for stochastic linear systems synchronization can be used in our framework to resolve the problem of distributed estimation. Here, we highlight that the results in this paper can be readily applied for achieving the synchronization of a class of stochastic linear systems with event-triggered communication schemes. 

To be concrete, let us consider a group of $m$ agents. The dynamics of each agent $i$ is given by a stochastic linear system as below:
\begin{equation}\label{eqn:synsys}
\nu_i(k+1) = \widetilde A \nu_i(k)+\widetilde Bu_i(k)+\widetilde L_i \omega_i(k),
\end{equation} 
where $\nu_i(k)$ and $u_i(k)$ are respectively the state and control input of the $i$-th agent, and $\omega_i(k)$ is the system noise with zero mean and bounded covariance. Then under a mild assumption that $(\widetilde A ,\widetilde B )$ is controllable, one can design the event-based controller as
\begin{equation}
u_i(k) =\widetilde\Gamma \sum_{j=1}^m a_{ij}(\widehat{\nu}_{j}(k)-\widehat{\nu}_{i}(k)),
\end{equation}
where $\widetilde\Gamma$ is determined by 
\begin{equation}
\widetilde\Gamma=\frac{2}{\mu_2+\mu_m}\frac{\widetilde B^T\mathcal{P}\widetilde{A}}{\widetilde B^T\mathcal{P}\widetilde B}.
\end{equation}
By doing so, the synchronization among agents is reached in the mean square sense with bounded error covariance. Namely, the consistency and consensus conditions \eqref{eqn:consistency2}--\eqref{eqn:consensus2} are reached. 


In these years, synchronization of stochastic linear systems with event-triggered schemes has received particular research attention. For example, Ma \textit{et al.} \cite{ma2016event} have focused on systems where the dynamics of each agent is subject to mutually uncorrelated zero-mean Gaussian white noises. Using linear matrix inequalities, they provide an event-based controller which facilitates the synchronization among agents in the mean square sense. Considering state-dependent noises, the authors of \cite{ding2015event} have leveraged the theory of input-to-state stability in probability and have derived sufficient conditions under which synchronization in probability is reached by using an
event-triggered control protocol. 

Different from the existing works, \eqref{eqn:synsys} deals with a more general class of noises. Hence, it includes the independent Gaussian white noise model in \cite{ma2016event} and the state-dependent noise model in \cite{ding2015event}. Specifically, the noises are only assumed to be bounded in covariance while they might be correlated with the states of agents along time and among agents. Because of its generality, our approach can be useful for various applications in both theoretical and engineering fields. 

\section{Low Message Complexity Estimator Design}\label{sec:Kdesign}
Noting that in practice, a communication channel in the sensor network is usually limited by a finite bandwidth, we finally investigate the design of distributed estimators under the constraint of message complexity. Specifically, suppose that the message complexity that the network is willing to tolerate is $\tilde{r}>0$. In this section, we show how to design the distributed estimator such that each sensor only sends messages of size no greater than $\tilde{r}$ when triggered.

To begin with, notice that any centralized Luenberger observer for estimating system \eqref{eqn:plant} is given by
\begin{equation}\label{eqn:Luenberger}
\hat{x}(k+1) =(A-K_{\tilde{r}}CA)\hat{x}(k)+K_{\tilde{r}}y(k+1),
\end{equation} 
where $K_{\tilde r}$ is the estimation gain of the Luenberger observer. Following similar arguments as in Sections~\ref{sec:decompose} and \ref{sec:implement}, it is not difficult to verify that Algorithm~\ref{alg:dist_est} can be generalized, by replacing $K$ with $K_{\tilde{r}}$, to achieve a distributed implementation of \eqref{eqn:Luenberger}. On the other hand, as stated in Remark~\ref{rmk:rankK}, the message complexity of implementing this distributed estimator is $\rank(K_{\tilde r})$. Hence, one way to reduce the message complexity is using a Luenberger observer with a low-rank estimation gain matrix, namely, $\rank(K_{\tilde r})\leq \tilde r$. Then by implementing it with Algorithm~\ref{alg:dist_est}, at each triggering instant, sensors only transmit a vector of dimension no more than $\tilde r$, which meets the network requirement. 

In what follows, we show how to design the optimal estimation gain $K_{\tilde r}$ in the sense that the Luenberger observer yields the minimum performance loss, under the constraint that $\rank(K_{\tilde r})\leq \tilde{r}$. To this end, let us factorize the estimation gain as 
	\begin{equation}\label{eqn:Kfactorize}
	K_{\tilde r}=\bar{K} W,
	\end{equation}
where $\bar{K}\in\mathbb{R}^{n\times \tilde{r}}$ and $W\in\mathbb{R}^{\tilde{r}\times m}$. We first consider how to design the optimal $\bar K$ when $W$ is given. After that, a semi-definite programming (SDP) is presented to compute the optimal $W$ under the constraint that $\rank(W)=\tilde{r}$. By virtue of \eqref{eqn:Kfactorize}, we conclude $\rank(K_{\tilde r})\leq \tilde{r}$.

\subsection{Optimal $\bar K$ when $W$ is given}
First, suppose that $W$ is given. We consider the following measurements given by a ``virtual" sensor network:
\begin{equation}\label{eqn:virtual_system}
\begin{split}
	\tilde{y}(k) = \tilde{C}x(k)+\tilde{v}(k),
\end{split}
\end{equation}
where 
\begin{equation}\label{eqn:virtualC}
\tilde{y}(k)=Wy(k),\; \tilde{C}=WC,\; \tilde{v}(k)=Wv(k).
\end{equation}
Suppose that this ``virtual" sensor network is monitoring the system \eqref{eqn:plant} and a Luenberger observer is performed with estimation gain $\bar K$, where $\bar K$ is defined in \eqref{eqn:Kfactorize}. Let us respectively denote by $\tilde{x}(k)$ and $\tilde{P}$ the corresponding estimate and error covariance. That is, 
\begin{equation}
\tilde{x}(k+1) = (A-\bar{K}\tilde{C}A)\tilde{x}(k)+\bar{K}\tilde{y}(k+1),
\end{equation}
and 
\begin{equation}
\tilde{P}(k) = \cov(\tilde{x}(k)-x(k)),\; \tilde{P}=\lim _{k \rightarrow \infty} \tilde{P}(k).
\end{equation}
The following result is immediate:
\begin{lemma}\label{lmm:tildeP}
Let $P_{\tilde r}$ be the estimation error covariance of the Luenberger observer~\eqref{eqn:Luenberger}.
Then it follows that
	\begin{equation}
		P_{\tilde r}=\tilde{P}.
	\end{equation}
\end{lemma}

As a result of Lemma~\ref{lmm:tildeP}, we would focus on finding the optimal $\bar{K}$ which minimizes $\tr(\tilde P)$. Clearly, for any given $W$, the optimal solution is provided by the Kalman filter, and the steady-state error covariance can be calculated as
\begin{equation}\label{eqn:optimalP}
	\tilde{P} = [(A\tilde{P}A^T+Q)^{-1}+\tilde{C}^T(\tilde{R})^{-1}\tilde C]^{-1},
\end{equation} 
where
\begin{equation}\label{eqn:tildeR}
\tilde{R}=WRW^T.
\end{equation}
Moreover, the optimal gain is given by
\begin{align}
	\bar{K} = (A\tilde{P}A^T+Q)\tilde C^T[\tilde{C}(A\tilde{P}A^T+Q)\tilde C^T+\tilde{R}]^{-1}. \label{eqn:optimalK}
\end{align}

\subsection{Towards finding an optimal $W$}
As seen from \eqref{eqn:optimalP}, the error covariance $\tilde P$ is a function of $W$.
Therefore, we next aim to find the optimal $W$ in the sense that $\tr(\tilde P)$ is minimized under the constraint that $\rank(W)=\tilde{r}$. Notice that $W$ only appears in the term $\tilde{C}^T(\tilde{R})^{-1}\tilde C$ of \eqref{eqn:optimalP}. We thus rewrite it as
\begin{equation}
\begin{aligned}
	&\tilde{C}^T(\tilde{R})^{-1}\tilde C = C^TW^T(WRW^T)^{-1}WC\\
	&= C^TR^{-1/2}[R^{1/2} W^T(WRW^T)^{-1}WR^{1/2}]R^{-1/2}C.
\end{aligned}
\end{equation} 
Let us denote
\begin{equation}\label{eqn:X}
	X\triangleq R^{1/2} W^T(WRW^T)^{-1}WR^{1/2}\in\mathbb{R}^{m\times m},\; \bar{C}\triangleq R^{-1/2}C.
\end{equation}
It is easy to verify that $X$ is a symmetric projection matrix, namely $X^2=X$ and $X=X^T$. Moreover, $\rank(X)=\rank(W)=\tilde{r}.$
On the other hand, given any symmetric projection matrix $X$ which is of rank $\tilde{r}$, one can always find $W$ that satisfies \eqref{eqn:X} by
\begin{equation}\label{eqn:W}
W= \Big(R^{-1/2}\begin{bmatrix}
v_1&\cdots& v_m
\end{bmatrix}\Big)^T,
\end{equation}
where $\{v_1, \cdots, v_m\}$ is the orthonormal basis of the column space of $X$. Therefore, instead of minimizing $\tr(\tilde P)$ over $W$, we can minimize it over $X$. Since the constraint on rank is not convex, we follow the approach in \cite{yuan2015security} and compute $X$ based on a convex relaxation by solving the following SDP:
 \begin{equation}\label{eqn:optimization2}
\begin{aligned}
	&\mathop{\textrm{minimize}}\limits_{X,\,\tilde P,\,\Theta}&
	& \tr(\tilde P) \\
	&\textrm{subject to}&
	&\begin{bmatrix}
		\tilde P&I\\
		I&\Theta
	\end{bmatrix}\geq 0, \\
	&&
	&\begin{bmatrix}
		Q^{-1}-\Theta+\bar{C}^TX\bar{C}& Q^{-1}A\\
		A^TQ^{-1}& \Theta+A^TQ^{-1}A
	\end{bmatrix} \geq 0 ,\\
	&&
	&X^T = X,\,0\leq X \leq I_m,\,\tr(X) = \tilde r. 
\end{aligned}
 \end{equation}

\begin{remark}
Notice that if $\tilde{r}>0$ and $(A,C)$ is observable, \eqref{eqn:optimization2} is always solvable, since one can verify that $X=\frac{\tilde{r}}{m} I_m$ is a feasible solution of it.
\end{remark}

For the problem \eqref{eqn:optimization2}, we can 
obtain the optimal solution $X_{*}$ and $\tilde{P}_{*}$. However, since the constraint on the rank of $X$ has been relaxed, the matrix $X_{*}$ is not necessarily a projection with rank $\tilde{r}$. In this case, one can obtain an approximation based on $X_*$. Specifically, we apply an eigendecomposion to $X_*$ as
$$
X_{*}=U_{*} \operatorname{diag}\left(\lambda_{1}, \ldots, \lambda_{p}\right) U_{*}^{T},
$$
where $U_{*}$ is orthonormal and $\lambda_{1} \geq \cdots \geq \lambda_{p}$ are the eigenvalues of $X_*$. We thus can obtain a projection matrix $X_{0}$ as
$$
X_{0}=U_{*} \operatorname{diag}(\underbrace{1, \ldots, 1}_{\tilde{r}}, \underbrace{0, \ldots, 0}_{m-\tilde{r}}) U_{*}^{T} .
$$
It is easy to verify $\tr(X_0)=\tilde{r}$ and thus the (sub)optimal $W$ is obtained by \eqref{eqn:W}. Then combining \eqref{eqn:Kfactorize}, \eqref{eqn:optimalP}, and \eqref{eqn:optimalK}, we finally obtain a (sub)optimal estimation gain 
$K_{\tilde r}$, the rank of which is no more than $\tilde{r}$. As discussed previously, by performing Algorithm~\ref{alg:dist_est}, one can, in a distributed manner, implement the Luenberger observer with gain $K_{\tilde r}$, where the message complexity is at most $\tilde{r}$. 

\begin{remark}
	As observed from \eqref{eqn:S}, the dimensions of $H$ and $L$ increase linearly with $\tilde{r}$. Therefore, using a low-rank estimator also helps to reduce the computation burden of sensors when performing the distributed estimation algorithm and making updates via \eqref{eqn:update}. This is especially beneficial if sensors are powered by energy-limited batteries.
\end{remark}


\section{Numerical Examples}\label{sec:simulation}
In this section, we would verify the established results through some numerical examples.

\subsection{Example $1$}
The first example aims to estimate the state of a $2$-dimensional LTI system with parameters
\begin{equation}
\begin{split}
A = 
\begin{bmatrix}
	0.9 & 0\\
	0 & 1.1
\end{bmatrix},\; Q=0.5I_2,
\end{split}
\end{equation}
To the end, $m=4$ sensors are deployed in the network. The measurements of them are given by 
\begin{equation}
		C = 
		\begin{bmatrix}
			1 & 0 & 1 & 1\\
			0 & 1 & 1 & -1
		\end{bmatrix}^T,
	\;R=2I_4.
\end{equation}
Suppose that the sensors are connected as a ring with $a_{ij}=1$ at each edge $(i,j)$. 
Then, it can be checked that $\mu_2=2$ and $\mu_4=4$. Moreover, we set the initial state $x(0)\sim \mathcal{N}(0,I)$, and choose $\zeta=0.5$ which meets the sufficient condition in Lemma~\ref{lmm:eta}.


In this example, \eqref{eqn:trigger1} is selected as the triggering function, where the parameters are set as $c_0=c_1=5$, $\alpha = 0.8$.
By performing Algorithm~\ref{alg:dist_est}, it is observed from Fig.~\ref{fig:x} that the mean square estimation error from each sensor is stable during the operation.  Moreover, Fig.~\ref{fig:performance} and Fig.~\ref{fig:performance2} show the box and whisker diagram obtained by the $1000$-run Monte Carlo trials, where the bottom and top of the box represent the first and third quartiles, the (red) band inside the box represents the median of the data, and the ends of the whiskers represent the minimum and maximum of the data. As illustrated, by choosing $c_0=c_1=5$, $\alpha = 0.8$, the average communication rate over the whole network is
$72.5\%$. On the other hand, compared to the distributed estimator in full transmission case \cite{yan2021distributed}, the estimation error incurred by the event-triggering mechanism is respectively $22.4\%, 8.8\%, 15.3\%,9.6\% $ for the four agents.  
Therefore, the proposed distributed estimator reduces the data transmission while preserving the estimation performance. Moreover, by increasing $c_0$ and $c_1$ to $8$, the communication among agents becomes less frequent, while leading to larger estimation error.

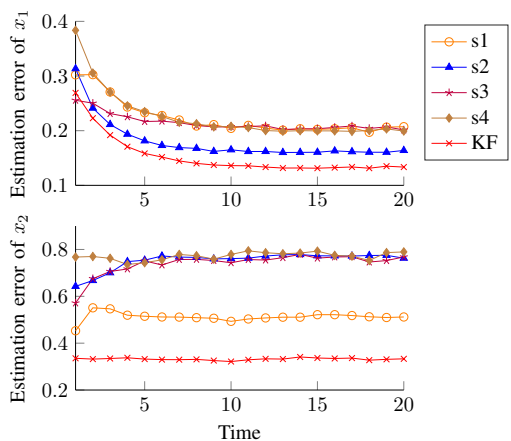
\begin{figure}[htbp]
	\centering
	\resizebox{0.37\textwidth}{!}{
%
%
\definecolor{mycolor1}{rgb}{0.00000,0.44700,0.74100}%
\definecolor{mycolor2}{rgb}{0.85000,0.32500,0.09800}%
\definecolor{mycolor3}{rgb}{0.92900,0.69400,0.12500}%
\definecolor{mycolor4}{rgb}{0.49400,0.18400,0.55600}%
\definecolor{mycolor5}{rgb}{0.46600,0.67400,0.18800}%
\begin{tikzpicture}

\begin{axis}[%
width=2.2in,
height=1.1in,
at={(0.758in,3.554in)},
scale only axis,
xmin=1,
xmax=20,
ymin=0.1,
ymax=0.4,
ylabel={Estimation error of $x_1$},
axis background/.style={fill=white},
axis x line*=bottom,
axis y line*=left
]
\addplot[color={orange}, mark={o}]
table[row sep=crcr]{%
1	0.302121180752325\\
2	0.302531925382281\\
3	0.270790892929535\\
4	0.243276799241696\\
5	0.232888346243705\\
6	0.227708282124915\\
7	0.219562708018138\\
8	0.208534774177071\\
9	0.211053634882657\\
10	0.204244486403274\\
11	0.209712530724859\\
12	0.204767961390524\\
13	0.20101659775227\\
14	0.203540896279591\\
15	0.202400557591565\\
16	0.203549485974262\\
17	0.205386769064392\\
18	0.19706440684058\\
19	0.20651303875293\\
20	0.207283075552277\\
};
\addplot[color={blue}, mark={triangle*}]
table[row sep=crcr]{%
1	0.313395924557476\\
2	0.241235028228442\\
3	0.211318017339914\\
4	0.193582756586918\\
5	0.181398105722853\\
6	0.172939605034201\\
7	0.169074571838835\\
8	0.167658191774771\\
9	0.162053809352254\\
10	0.165007345486911\\
11	0.161879212916129\\
12	0.161951350345411\\
13	0.160099139842983\\
14	0.160295758177786\\
15	0.160554931811933\\
16	0.163059177764199\\
17	0.161532948376557\\
18	0.16056133123955\\
19	0.160526809138623\\
20	0.163817507865157\\
};
\addplot[color={purple}, mark={star}]
table[row sep=crcr]{%
1	0.255628916441544\\
2	0.250275356313448\\
3	0.231336405929384\\
4	0.225384695644437\\
5	0.216804303447242\\
6	0.217448469635608\\
7	0.214714453980911\\
8	0.209494214895698\\
9	0.206950355341257\\
10	0.207044763187027\\
11	0.207671198596366\\
12	0.208831784760228\\
13	0.202192460601073\\
14	0.203233278598226\\
15	0.203549073895409\\
16	0.206022226570817\\
17	0.208622191746696\\
18	0.204092133426298\\
19	0.206947620037054\\
20	0.20132985249283\\
};
\addplot[color={brown}, mark={diamond*}]
table[row sep=crcr]{%
1	0.383593560200717\\
2	0.305409921380758\\
3	0.270690727454749\\
4	0.245413374750699\\
5	0.234761145971102\\
6	0.225984523706989\\
7	0.21538058064133\\
8	0.212691680438308\\
9	0.20697878705572\\
10	0.207942236185892\\
11	0.205738125212475\\
12	0.200243252998125\\
13	0.199030670595392\\
14	0.20004087672674\\
15	0.199644087660341\\
16	0.199642238130119\\
17	0.198791611562973\\
18	0.20033477136046\\
19	0.203273667574249\\
20	0.199362563701373\\
};
\addplot[color={red}, mark={x}]
table[row sep=crcr]{%
1	0.269151387318479\\
2	0.222787204893859\\
3	0.191867442683382\\
4	0.17071687357105\\
5	0.158362032512998\\
6	0.151740557392187\\
7	0.144694752495978\\
8	0.140171075377349\\
9	0.137224529367621\\
10	0.136089865134629\\
11	0.135768715821108\\
12	0.133440549024216\\
13	0.131664391376833\\
14	0.131750764538463\\
15	0.131352076748205\\
16	0.132451174818388\\
17	0.133625149467611\\
18	0.131341925245552\\
19	0.135041594288518\\
20	0.133428591892059\\
};
\end{axis}

\begin{axis}[%
	width=2.2in,
	height=1.1in,
	at={(0.758in,2.181in)},
	scale only axis,
	xmin=1,
	xmax=20,
	xlabel={Time},
	ymin=0.2,
	ymax=0.9,
	ylabel={Estimation error of $x_2$},
	axis background/.style={fill=white},
	axis x line*=bottom,
	axis y line*=left,
	legend style={at={(1.063,1.40)}, anchor=south west, legend cell align=left, align=left, draw=white!15!black}
	]
\addplot[color={orange}, mark={o}]
table[row sep=crcr]{%
1	0.453029400067562\\
2	0.550894916164929\\
3	0.546583406464047\\
4	0.519424517870761\\
5	0.514885188224169\\
6	0.51180799840207\\
7	0.511305994808284\\
8	0.508836527397316\\
9	0.506342619585707\\
10	0.493212400205813\\
11	0.503271154377489\\
12	0.507700269914462\\
13	0.51092684860295\\
14	0.510913991263884\\
15	0.521600196367656\\
16	0.521611587238072\\
17	0.517898687406905\\
18	0.512692699662478\\
19	0.509555786187111\\
20	0.511801907551975\\
};
\addlegendentry {s1}
\addplot[color={blue}, mark={triangle*}]
table[row sep=crcr]{%
1	0.64256224120478\\
2	0.667407709595111\\
3	0.700505363251635\\
4	0.748112559935401\\
5	0.75429870078848\\
6	0.772128134221469\\
7	0.76752966610019\\
8	0.766036084665164\\
9	0.759784539467824\\
10	0.76093452019309\\
11	0.763007790740956\\
12	0.771904091392439\\
13	0.777281953590168\\
14	0.778221284374065\\
15	0.7730066715847\\
16	0.772518442758351\\
17	0.771820759540536\\
18	0.774208314814909\\
19	0.774978089237367\\
20	0.762765920394484\\
};
\addlegendentry {s2}
\addplot[color={purple}, mark={star}]
table[row sep=crcr]{%
1	0.571474875195974\\
2	0.675312231042555\\
3	0.705984121766989\\
4	0.716775804251697\\
5	0.751121363859044\\
6	0.73382704125233\\
7	0.757146966662187\\
8	0.75768102691019\\
9	0.752797640999932\\
10	0.743258047374474\\
11	0.756937001846725\\
12	0.754813862167644\\
13	0.765013420825248\\
14	0.778874823587747\\
15	0.762973448937797\\
16	0.767879598031928\\
17	0.769088959247574\\
18	0.746368272005167\\
19	0.752541820631393\\
20	0.768087777045378\\
};
\addlegendentry {s3}
\addplot[color={brown}, mark={diamond*}]
table[row sep=crcr]{%
1	0.767604637138717\\
2	0.769929305891828\\
3	0.762157060919805\\
4	0.735767835867393\\
5	0.743512110918758\\
6	0.755833278625771\\
7	0.778664236968829\\
8	0.773203067561396\\
9	0.758610722451041\\
10	0.778752617450897\\
11	0.794174110672067\\
12	0.786078575600135\\
13	0.781627873057487\\
14	0.784584047550623\\
15	0.79282145228034\\
16	0.774800854703678\\
17	0.771372078508285\\
18	0.753077713435502\\
19	0.786568315416121\\
20	0.789581219951017\\
};
\addlegendentry { s4}
\addplot[color={red}, mark={x}]
table[row sep=crcr]{%
1	0.335392622620581\\
2	0.332043111066769\\
3	0.334597545474621\\
4	0.337678400104026\\
5	0.332057650567736\\
6	0.329857514633496\\
7	0.329582667601162\\
8	0.330800294478031\\
9	0.325731003810904\\
10	0.321566062378603\\
11	0.329067603414451\\
12	0.333154316076968\\
13	0.332134809135495\\
14	0.341081941775255\\
15	0.336722952619112\\
16	0.334685507834207\\
17	0.336416876931158\\
18	0.327404518189118\\
19	0.330965011755791\\
20	0.333209344369545\\
};
\addlegendentry {KF}
\end{axis}
\end{tikzpicture}%
}
	\caption{Average mean square estimation error of system states in $1000$-run Monte Carlo trials.}
	\label{fig:x}
\end{figure}

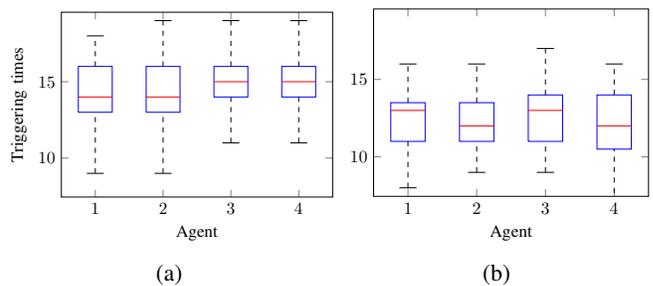
\begin{figure}[htbp]
	\centering
	\begin{minipage}[b]{0.24\textwidth}
		\centering
		\resizebox{\textwidth}{!}{
%
%
\begin{tikzpicture}

\begin{axis}[%
width=2.2in,
height=1.5in,
at={(1.011in,0.642in)},
scale only axis,
unbounded coords=jump,
xmin=0.5,
xmax=4.5,
xtick={1,2,3,4},
xlabel = { Agent},
ymin=7.45,
ymax=19.55,
ylabel = {Triggering times},
axis background/.style={fill=white}
]
\addplot [color=black, dashed, forget plot]
  table[row sep=crcr]{%
1	16\\
1	18\\
};
\addplot [color=black, dashed, forget plot]
  table[row sep=crcr]{%
2	16\\
2	19\\
};
\addplot [color=black, dashed, forget plot]
  table[row sep=crcr]{%
3	16\\
3	19\\
};
\addplot [color=black, dashed, forget plot]
  table[row sep=crcr]{%
4	16\\
4	19\\
};
\addplot [color=black, dashed, forget plot]
  table[row sep=crcr]{%
1	9\\
1	13\\
};
\addplot [color=black, dashed, forget plot]
  table[row sep=crcr]{%
2	9\\
2	13\\
};
\addplot [color=black, dashed, forget plot]
  table[row sep=crcr]{%
3	11\\
3	14\\
};
\addplot [color=black, dashed, forget plot]
  table[row sep=crcr]{%
4	11\\
4	14\\
};
\addplot [color=black, forget plot]
  table[row sep=crcr]{%
0.875	18\\
1.125	18\\
};
\addplot [color=black, forget plot]
  table[row sep=crcr]{%
1.875	19\\
2.125	19\\
};
\addplot [color=black, forget plot]
  table[row sep=crcr]{%
2.875	19\\
3.125	19\\
};
\addplot [color=black, forget plot]
  table[row sep=crcr]{%
3.875	19\\
4.125	19\\
};
\addplot [color=black, forget plot]
  table[row sep=crcr]{%
0.875	9\\
1.125	9\\
};
\addplot [color=black, forget plot]
  table[row sep=crcr]{%
1.875	9\\
2.125	9\\
};
\addplot [color=black, forget plot]
  table[row sep=crcr]{%
2.875	11\\
3.125	11\\
};
\addplot [color=black, forget plot]
  table[row sep=crcr]{%
3.875	11\\
4.125	11\\
};
\addplot [color=blue, forget plot]
  table[row sep=crcr]{%
0.75	13\\
0.75	16\\
1.25	16\\
1.25	13\\
0.75	13\\
};
\addplot [color=blue, forget plot]
  table[row sep=crcr]{%
1.75	13\\
1.75	16\\
2.25	16\\
2.25	13\\
1.75	13\\
};
\addplot [color=blue, forget plot]
  table[row sep=crcr]{%
2.75	14\\
2.75	16\\
3.25	16\\
3.25	14\\
2.75	14\\
};
\addplot [color=blue, forget plot]
  table[row sep=crcr]{%
3.75	14\\
3.75	16\\
4.25	16\\
4.25	14\\
3.75	14\\
};
\addplot [color=red, forget plot]
  table[row sep=crcr]{%
0.75	14\\
1.25	14\\
};
\addplot [color=red, forget plot]
  table[row sep=crcr]{%
1.75	14\\
2.25	14\\
};
\addplot [color=red, forget plot]
  table[row sep=crcr]{%
2.75	15\\
3.25	15\\
};
\addplot [color=red, forget plot]
  table[row sep=crcr]{%
3.75	15\\
4.25	15\\
};
\end{axis}

\end{tikzpicture}%
}
		\subcaption{}
	\end{minipage}%
	\begin{minipage}[b]{0.24\textwidth}
	\centering
	\resizebox{0.93\textwidth}{!}{
%
%
\begin{tikzpicture}

\begin{axis}[%
width=2.2in,
height=1.5in,
at={(1.011in,0.642in)},
scale only axis,
unbounded coords=jump,
xmin=0.5,
xmax=4.5,
xtick={1,2,3,4},
xlabel = { Agent},
ymin=7.45,
ymax=19.55,
ylabel = {},
axis background/.style={fill=white}
]
\addplot [color=black, dashed, forget plot]
  table[row sep=crcr]{%
1	13.5\\
1	16\\
};
\addplot [color=black, dashed, forget plot]
  table[row sep=crcr]{%
2	13.5\\
2	16\\
};
\addplot [color=black, dashed, forget plot]
  table[row sep=crcr]{%
3	14\\
3	17\\
};
\addplot [color=black, dashed, forget plot]
  table[row sep=crcr]{%
4	14\\
4	16\\
};
\addplot [color=black, dashed, forget plot]
  table[row sep=crcr]{%
1	8\\
1	11\\
};
\addplot [color=black, dashed, forget plot]
  table[row sep=crcr]{%
2	9\\
2	11\\
};
\addplot [color=black, dashed, forget plot]
  table[row sep=crcr]{%
3	9\\
3	11\\
};
\addplot [color=black, dashed, forget plot]
  table[row sep=crcr]{%
4	7\\
4	10.5\\
};
\addplot [color=black, forget plot]
  table[row sep=crcr]{%
0.875	16\\
1.125	16\\
};
\addplot [color=black, forget plot]
  table[row sep=crcr]{%
1.875	16\\
2.125	16\\
};
\addplot [color=black, forget plot]
  table[row sep=crcr]{%
2.875	17\\
3.125	17\\
};
\addplot [color=black, forget plot]
  table[row sep=crcr]{%
3.875	16\\
4.125	16\\
};
\addplot [color=black, forget plot]
  table[row sep=crcr]{%
0.875	8\\
1.125	8\\
};
\addplot [color=black, forget plot]
  table[row sep=crcr]{%
1.875	9\\
2.125	9\\
};
\addplot [color=black, forget plot]
  table[row sep=crcr]{%
2.875	9\\
3.125	9\\
};
\addplot [color=black, forget plot]
  table[row sep=crcr]{%
3.875	7\\
4.125	7\\
};
\addplot [color=blue, forget plot]
  table[row sep=crcr]{%
0.75	11\\
0.75	13.5\\
1.25	13.5\\
1.25	11\\
0.75	11\\
};
\addplot [color=blue, forget plot]
  table[row sep=crcr]{%
1.75	11\\
1.75	13.5\\
2.25	13.5\\
2.25	11\\
1.75	11\\
};
\addplot [color=blue, forget plot]
  table[row sep=crcr]{%
2.75	11\\
2.75	14\\
3.25	14\\
3.25	11\\
2.75	11\\
};
\addplot [color=blue, forget plot]
  table[row sep=crcr]{%
3.75	10.5\\
3.75	14\\
4.25	14\\
4.25	10.5\\
3.75	10.5\\
};
\addplot [color=red, forget plot]
  table[row sep=crcr]{%
0.75	13\\
1.25	13\\
};
\addplot [color=red, forget plot]
  table[row sep=crcr]{%
1.75	12\\
2.25	12\\
};
\addplot [color=red, forget plot]
  table[row sep=crcr]{%
2.75	13\\
3.25	13\\
};
\addplot [color=red, forget plot]
  table[row sep=crcr]{%
3.75	12\\
4.25	12\\
};
\end{axis}
\end{tikzpicture}%
}
	\subcaption{}
	\label{fig:times2}
\end{minipage}
	\caption{Triggering times of each agent in $20$ iterations by performing Algorithm~\ref{alg:dist_est}, where the parameters of \eqref{eqn:trigger1} are (a) $c_0=c_1=5$, $\alpha = 0.8$; (b) $c_0=c_1=8$, $\alpha = 0.8$.}
\label{fig:performance}
\end{figure}
\begin{figure}[h!]
	\begin{minipage}[b]{0.24\textwidth}
		\centering
		\resizebox{\textwidth}{!}{
%
%
\begin{tikzpicture}

\begin{axis}[%
width=2.2in,
height=1.5in,
at={(2.6in,0.811in)},
scale only axis,
xmin=0.5,
xmax=4.5,
xtick={1,2,3,4},
xlabel = {Agent},
ymin=0.8,
ymax=2.5,
ylabel = {Relative error},
axis background/.style={fill=white}
]
\addplot [color=black, dashed, forget plot]
  table[row sep=crcr]{%
1	1.26184660663355\\
1	1.3384458823965\\
};
\addplot [color=black, dashed, forget plot]
  table[row sep=crcr]{%
2	1.0991783823695\\
2	1.12075417553215\\
};
\addplot [color=black, dashed, forget plot]
  table[row sep=crcr]{%
3	1.1638939857908\\
3	1.21655761236679\\
};
\addplot [color=black, dashed, forget plot]
  table[row sep=crcr]{%
4	1.1108622728191\\
4	1.14487752044923\\
};
\addplot [color=black, dashed, forget plot]
  table[row sep=crcr]{%
1	1.13276062863677\\
1	1.20937840343485\\
};
\addplot [color=black, dashed, forget plot]
  table[row sep=crcr]{%
2	1.02319120318998\\
2	1.06849474167272\\
};
\addplot [color=black, dashed, forget plot]
  table[row sep=crcr]{%
3	1.07430775422801\\
3	1.12802694994855\\
};
\addplot [color=black, dashed, forget plot]
  table[row sep=crcr]{%
4	1.07269971029293\\
4	1.08776365071064\\
};
\addplot [color=black, forget plot]
  table[row sep=crcr]{%
0.875	1.3384458823965\\
1.125	1.3384458823965\\
};
\addplot [color=black, forget plot]
  table[row sep=crcr]{%
1.875	1.12075417553215\\
2.125	1.12075417553215\\
};
\addplot [color=black, forget plot]
  table[row sep=crcr]{%
2.875	1.21655761236679\\
3.125	1.21655761236679\\
};
\addplot [color=black, forget plot]
  table[row sep=crcr]{%
3.875	1.14487752044923\\
4.125	1.14487752044923\\
};
\addplot [color=black, forget plot]
  table[row sep=crcr]{%
0.875	1.13276062863677\\
1.125	1.13276062863677\\
};
\addplot [color=black, forget plot]
  table[row sep=crcr]{%
1.875	1.02319120318998\\
2.125	1.02319120318998\\
};
\addplot [color=black, forget plot]
  table[row sep=crcr]{%
2.875	1.07430775422801\\
3.125	1.07430775422801\\
};
\addplot [color=black, forget plot]
  table[row sep=crcr]{%
3.875	1.07269971029293\\
4.125	1.07269971029293\\
};
\addplot [color=blue, forget plot]
  table[row sep=crcr]{%
0.75	1.20937840343485\\
0.75	1.26184660663355\\
1.25	1.26184660663355\\
1.25	1.20937840343485\\
0.75	1.20937840343485\\
};
\addplot [color=blue, forget plot]
  table[row sep=crcr]{%
1.75	1.06849474167272\\
1.75	1.0991783823695\\
2.25	1.0991783823695\\
2.25	1.06849474167272\\
1.75	1.06849474167272\\
};
\addplot [color=blue, forget plot]
  table[row sep=crcr]{%
2.75	1.12802694994855\\
2.75	1.1638939857908\\
3.25	1.1638939857908\\
3.25	1.12802694994855\\
2.75	1.12802694994855\\
};
\addplot [color=blue, forget plot]
  table[row sep=crcr]{%
3.75	1.08776365071064\\
3.75	1.1108622728191\\
4.25	1.1108622728191\\
4.25	1.08776365071064\\
3.75	1.08776365071064\\
};
\addplot [color=red, forget plot]
  table[row sep=crcr]{%
0.75	1.22408639948477\\
1.25	1.22408639948477\\
};
\addplot [color=red, forget plot]
  table[row sep=crcr]{%
1.75	1.08807001658301\\
2.25	1.08807001658301\\
};
\addplot [color=red, forget plot]
  table[row sep=crcr]{%
2.75	1.1537660488257\\
3.25	1.1537660488257\\
};
\addplot [color=red, forget plot]
  table[row sep=crcr]{%
3.75	1.09672901594904\\
4.25	1.09672901594904\\
};
\end{axis}
\end{tikzpicture}%
}
		\subcaption{}
	\end{minipage}%
	\begin{minipage}[b]{0.24\textwidth}
	\centering
	\resizebox{0.93\textwidth}{!}{
%
%
\begin{tikzpicture}

\begin{axis}[%
width=2.2in,
height=1.5in,
at={(0.758in,0.481in)},
scale only axis,
xmin=0.5,
xmax=4.5,
xlabel = {Agent},
xtick={1,2,3,4},
ymin=0.8,
ymax=2.5,
axis background/.style={fill=white}
]
\addplot [color=black, dashed, forget plot]
  table[row sep=crcr]{%
1	2.16696927872065\\
1	2.39847102682907\\
};
\addplot [color=black, dashed, forget plot]
  table[row sep=crcr]{%
2	1.88442863399895\\
2	1.97454433907615\\
};
\addplot [color=black, dashed, forget plot]
  table[row sep=crcr]{%
3	1.81896076622833\\
3	1.88389512233802\\
};
\addplot [color=black, dashed, forget plot]
  table[row sep=crcr]{%
4	1.72565330233914\\
4	1.82471810877544\\
};
\addplot [color=black, dashed, forget plot]
  table[row sep=crcr]{%
1	1.88826027188363\\
1	1.94458757134078\\
};
\addplot [color=black, dashed, forget plot]
  table[row sep=crcr]{%
2	1.71622291110533\\
2	1.81708494294614\\
};
\addplot [color=black, dashed, forget plot]
  table[row sep=crcr]{%
3	1.70897201427988\\
3	1.77470189423576\\
};
\addplot [color=black, dashed, forget plot]
  table[row sep=crcr]{%
4	1.59106084479358\\
4	1.65904697774575\\
};
\addplot [color=black, forget plot]
  table[row sep=crcr]{%
0.875	2.39847102682907\\
1.125	2.39847102682907\\
};
\addplot [color=black, forget plot]
  table[row sep=crcr]{%
1.875	1.97454433907615\\
2.125	1.97454433907615\\
};
\addplot [color=black, forget plot]
  table[row sep=crcr]{%
2.875	1.88389512233802\\
3.125	1.88389512233802\\
};
\addplot [color=black, forget plot]
  table[row sep=crcr]{%
3.875	1.82471810877544\\
4.125	1.82471810877544\\
};
\addplot [color=black, forget plot]
  table[row sep=crcr]{%
0.875	1.88826027188363\\
1.125	1.88826027188363\\
};
\addplot [color=black, forget plot]
  table[row sep=crcr]{%
1.875	1.71622291110533\\
2.125	1.71622291110533\\
};
\addplot [color=black, forget plot]
  table[row sep=crcr]{%
2.875	1.70897201427988\\
3.125	1.70897201427988\\
};
\addplot [color=black, forget plot]
  table[row sep=crcr]{%
3.875	1.59106084479358\\
4.125	1.59106084479358\\
};
\addplot [color=blue, forget plot]
  table[row sep=crcr]{%
0.75	1.94458757134078\\
0.75	2.16696927872065\\
1.25	2.16696927872065\\
1.25	1.94458757134078\\
0.75	1.94458757134078\\
};
\addplot [color=blue, forget plot]
  table[row sep=crcr]{%
1.75	1.81708494294614\\
1.75	1.88442863399895\\
2.25	1.88442863399895\\
2.25	1.81708494294614\\
1.75	1.81708494294614\\
};
\addplot [color=blue, forget plot]
  table[row sep=crcr]{%
2.75	1.77470189423576\\
2.75	1.81896076622833\\
3.25	1.81896076622833\\
3.25	1.77470189423576\\
2.75	1.77470189423576\\
};
\addplot [color=blue, forget plot]
  table[row sep=crcr]{%
3.75	1.65904697774575\\
3.75	1.72565330233914\\
4.25	1.72565330233914\\
4.25	1.65904697774575\\
3.75	1.65904697774575\\
};
\addplot [color=red, forget plot]
  table[row sep=crcr]{%
0.75	2.02435903534918\\
1.25	2.02435903534918\\
};
\addplot [color=red, forget plot]
  table[row sep=crcr]{%
1.75	1.8509199514391\\
2.25	1.8509199514391\\
};
\addplot [color=red, forget plot]
  table[row sep=crcr]{%
2.75	1.79635925667301\\
3.25	1.79635925667301\\
};
\addplot [color=red, forget plot]
  table[row sep=crcr]{%
3.75	1.68625159802322\\
4.25	1.68625159802322\\
};
\end{axis}
\end{tikzpicture}%
}
	\subcaption{}
\end{minipage}%
\caption{Relative estimation error of each agent in comparison with the full transmission case by performing Algorithm~\ref{alg:dist_est}, where the parameters of \eqref{eqn:trigger1} are (a) $c_0=c_1=5$, $\alpha = 0.8$; (b) $c_0=c_1=8$, $\alpha = 0.8$.}
\label{fig:performance2}
\end{figure}
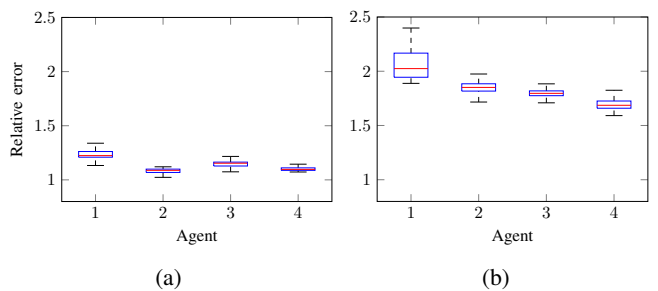


\subsection{Example $2$}
The second example considers a system of larger scale. Concretely, we study the heat transfer process in a square region as presented in \cite{yan2021distributed}. As shown in Fig.~\ref{fig:topology}, $m=15$ sensors are deployed for monitoring temperature within the region represented by a $5\times 5$ grid. The temperature of each grid is taken as a state and thus $n=25$. The detailed procedure for modeling the system is omitted here. However, interested readers can refer to \cite{yan2021distributed}. The covariance of system and measurement noises is chosen as $Q=I_{n}$ and $R=I_{m}$.
\begin{figure}[]
	\centering
	\includegraphics[width=0.23\textwidth]{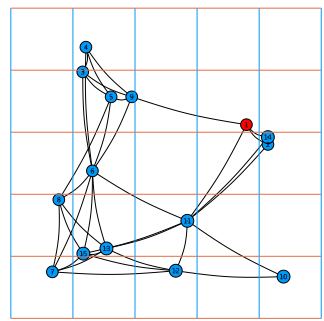}
	\caption{The location and topology of $m$ sensors in the grid.}
	\label{fig:topology}
\end{figure}

\subsubsection{Performance of low-rank estimators}
With the given system, we first illustrate the low-rank estimator as designed in Section~\ref{sec:Kdesign}. Fig.~\ref{fig:lowrank} shows the estimation error of the designed estimators with rank $\tilde r=1,5,7,15$, respectively. Note that $m=15$ is the original Kalman filter case. One can see that even if $\tilde r = 7$, meaning that we only use less than half of the degree of freedom 
to design the estimator, the performance loss is small as around $5\%$.
\begin{figure}[htbp]
	\centering
	\resizebox{0.42\textwidth}{!}{
%
%
\definecolor{mycolor1}{rgb}{0.00000,0.44700,0.74100}%
\definecolor{mycolor2}{rgb}{0.85000,0.32500,0.09800}%
\definecolor{mycolor3}{rgb}{0.92900,0.69400,0.12500}%
\definecolor{mycolor4}{rgb}{0.49400,0.18400,0.55600}%
\begin{tikzpicture}
	
	\begin{axis}[%
		width=2.2in,
		height=1.3in,
		at={(0in,0in)},
		scale only axis,
		xmin=1,
		xmax=30,
		xlabel={{Time}},
		ymin=20,
		ymax=50,
		ylabel={{ Estimation error}},
		axis background/.style={fill=white},
		axis x line*=bottom,
		axis y line*=left,
		legend style={legend cell align=left, align=left, draw=white!15!black},
		legend pos={outer north east}
		]
		\addplot [color=mycolor1, mark=o,mark size=1.5pt]
		table[row sep=crcr]{%
			1	24.7329749946281\\
			2	30.9535641489055\\
			3	34.2157692931679\\
			4	36.5547390365517\\
			5	38.0234873103829\\
			6	39.2170431160262\\
			7	40.3041990683507\\
			8	41.1810969876393\\
			9	42.0298804453696\\
			10	42.1993890044411\\
			11	42.5549174069529\\
			12	42.9879539347682\\
			13	43.3449334548898\\
			14	43.616577531512\\
			15	43.7212458590475\\
			16	43.8074324613189\\
			17	43.9971030445961\\
			18	44.4736162807069\\
			19	44.452857553664\\
			20	44.3008473571626\\
			21	44.3874158356297\\
			22	44.6647640530355\\
			23	44.764418466552\\
			24	44.8994120901792\\
			25	44.8994120901792\\
			26	44.8994120901792\\
			27	44.8994120901792\\
			28	44.8994120901792\\
			29	44.8994120901792\\
			30	44.8994120901792\\
		};
		\addlegendentry{{\footnotesize $\tilde r = 1$}}
		
		\addplot [color=mycolor2, mark=square,mark size=1.5pt]
		table[row sep=crcr]{%
			1	23.5005382895565\\
			2	28.0290733706933\\
			3	29.8940795277087\\
			4	30.6544495334603\\
			5	31.130334615928\\
			6	31.2833492930407\\
			7	31.3243312727608\\
			8	31.4035573792787\\
			9	31.7126826460809\\
			10	31.7103592520702\\
			11	31.7163107651764\\
			12	31.5458627554721\\
			13	31.6786936927278\\
			14	31.7101779913316\\
			15	31.4037625806342\\
			16	31.4822276255249\\
			17	31.6972788224781\\
			18	31.8233600907493\\
			19	31.7429906278426\\
			20	31.5981825494562\\
			21	31.5363825415017\\
			22	31.6693191208636\\
			23	31.5433451258478\\
			24	31.5958568401481\\
			25	31.84277280916\\
			26	31.610345943169\\
			27	31.655165806214\\
			28	31.6172456548625\\
			29	31.5468721224731\\
			30	31.527179642016\\
		};
		\addlegendentry{{\footnotesize $\tilde r = 5$}}
		
		\addplot [color=mycolor3, mark=x]
		table[row sep=crcr]{%
			1	22.6742272250259\\
			2	26.6608899448125\\
			3	27.9345214663098\\
			4	28.3677669559724\\
			5	28.8018359398739\\
			6	28.8390915411213\\
			7	29.0861361443086\\
			8	29.0391717248065\\
			9	29.0034513407184\\
			10	29.0983736225117\\
			11	29.111899327807\\
			12	29.1973251355601\\
			13	29.1555585426452\\
			14	29.121579977852\\
			15	29.1456579584492\\
			16	29.198608083605\\
			17	29.1800164394459\\
			18	29.121782023755\\
			19	29.1042548370255\\
			20	29.1815963395368\\
			21	29.138416412395\\
			22	29.1348297492142\\
			23	29.1641938969396\\
			24	29.159006179247\\
			25	29.1769226997966\\
			26	29.1690553518847\\
			27	28.9542096695196\\
			28	29.2394209165204\\
			29	29.1213578144527\\
			30	29.3227616894127\\
		};
		\addlegendentry{{\footnotesize $\tilde r =7$}}
		
		\addplot [color=mycolor4, mark=triangle]
		table[row sep=crcr]{%
			1	21.9187422969218\\
			2	25.4560238534798\\
			3	26.6809593841344\\
			4	27.1774975410364\\
			5	27.5530005648025\\
			6	27.6068638775199\\
			7	27.6464350000402\\
			8	27.6114641342622\\
			9	27.6015176022023\\
			10	27.8019977107359\\
			11	27.6277782044806\\
			12	27.8043428550983\\
			13	27.7971208938815\\
			14	27.770135326414\\
			15	27.7639251465458\\
			16	27.6533649228409\\
			17	27.6286998396839\\
			18	27.5230323871013\\
			19	27.6940985893231\\
			20	27.606567649928\\
			21	27.6181200494654\\
			22	27.8044927494749\\
			23	27.8048329294441\\
			24	27.7732428558875\\
			25	27.8904679202486\\
			26	27.9282483571417\\
			27	27.8129638395203\\
			28	27.8211947073575\\
			29	27.6845575753418\\
			30	27.7242241752132\\
		};
		\addlegendentry{{\footnotesize $\tilde r = 15$}}      
		
	\end{axis}
	
\end{tikzpicture}%
}
	\caption{Average mean square estimation error of low-rank estimators in $1000$-run Monte Carlo trials.}
	\label{fig:lowrank}
\end{figure}
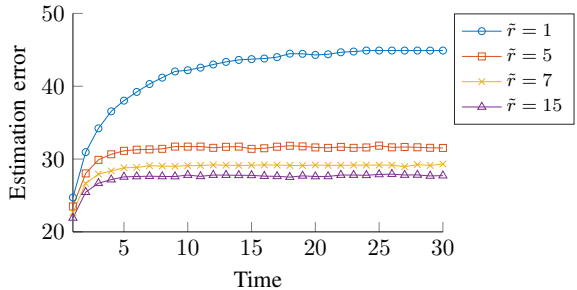

TABLE ~\ref{table:improvement} further compares the estimation performance of the low-rank estimator versus that of the Kalman filter with no rank constraint.  Given $\tilde r$ varying from $0$ to $15$, we define the relative performance as
\begin{equation}
J_{\tilde r} \triangleq \frac{\tr(P_{\tilde r})}{\tr(P)},
\end{equation}
where $P$ and $P_{\tilde r}$ are respectively the steady-state error covariances of the Kalman filter and the proposed estimator with rank $\tilde r$. 

\begin{table*}[]
	\small
	\caption{Relative Performance of Designed Estimator under Rank Constraint}
	\centering
	\label{table:improvement}
	\begin{tabular}{m{0.55cm}<{\centering} *{15}{m{0.55cm}<{\centering}}}
		\hline
		$\tilde r$
		& 1  &2 & 3 & 4 & 5 & 6 &7 & 8 & 9 & 10 & 11  & 12 & 13 & 14 & 15\\ \hline
	$J_{\tilde r}$ & 1.628    &1.442    &1.226   &1.141    &1.104    &1.071    &1.050    &1.033    &1.022    &1.014    &1.006    &1.002    &1.001 &1.001         &1.000\\\hline
	\end{tabular}
\end{table*}

\subsubsection{Performance of different algorithms}
The estimation performance of different algorithms is presented in Fig.~\ref{fig:heatmap}. By using an event-based communication strategy, Algorithm~\ref{alg:dist_est} inevitably incurs more estimation error than the centralized Kalman filter and the distributed estimator with full transmission, which always have the information of all sensors and all neighboring sensors, respectively. Moreover, we further implement the designed low-rank estimator, where $\tilde r =1$, in a distributed manner by replacing $K$ with $K_{\tilde r}$ in Algorithm~\ref{alg:dist_est}. As shown in Fig.~\ref{fig:heatmap}(d), it also yields a stable local estimate.
\begin{figure}[]
	\centering
	\includegraphics[width=0.47\textwidth]{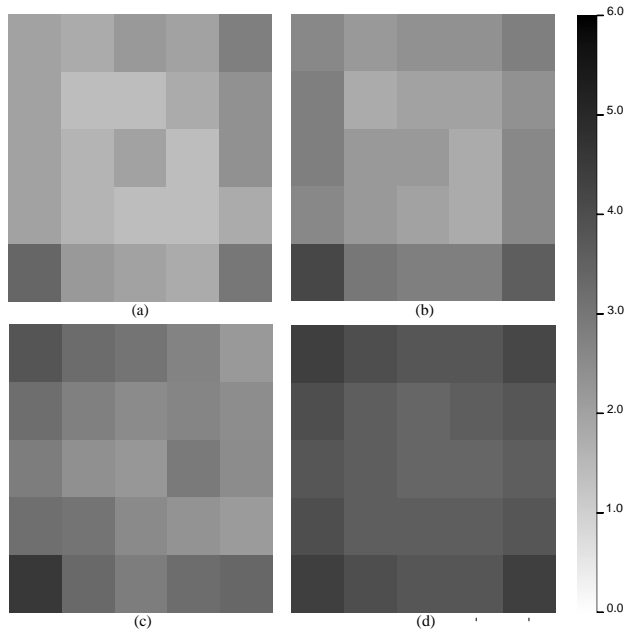}
	\caption{The squared estimation error of sensor $1$ by performing (a) centralized Kalman filter, (b) distributed implementation of centralized Kalman filter with full transmission (\hspace{1pt}\cite{yan2021distributed}), (c) distributed implementation of centralized Kalman filter with event-based transmission (Algorithm~\ref{alg:dist_est}), (d) distributed implementation of the designed rank-$1$ estimator with event-based transmission.}
	\label{fig:heatmap}
\end{figure}

\section{Conclusion}\label{sec:conclusion}
This paper has addressed the problem of distributed state estimation with event-based communication protocols. By decomposing the centralized estimator, we have reformulated the problem of distributed estimation to that of stochastic linear systems synchronization, in which a large class of triggering functions has been proved to be effective in yielding a stable local estimate at every sensor side. Given any $\tilde r$, an SDP is presented which gives the (sub)optimal gain of centralized estimator such that the distributed estimation algorithm can be implemented with message complexity no more than $\tilde{r}$, which is lower than that in the existing works. Finally, as we have discussed, the framework proposed in this paper can potentially be extended to achieve a distributed implementation of a stable Luenberger estimator under other noise models, e.g., $H_\infty$ estimator in the presence of bounded noise, which will be left as one of our future works.

\appendices

\section{Proof of Lemma~\ref{lmm:localeqv}}\label{sec:app_localeqv}

1) To begin with, it follows from \eqref{eqn:plant} and \eqref{eqn:diagA} that 
\begin{equation}\label{eqn:xs}
x^s(k+1) = A^s x^s(k)+Jw(k),
\end{equation}
where $J = \begin{bmatrix}
0 & \1_{n^s}
\end{bmatrix}\in\mathbb{R}^{n_s\times n}$ and $x(k) = \col(x^u(k), x^s(k)).$ Moreover, let us partition $C_i$ according to \eqref{eqn:diagA} as
\begin{equation}\label{eqn:diagC}
\begin{aligned}
C_i = \begin{bmatrix}
C_i^u & C_i^s
\end{bmatrix},
\end{aligned}
\end{equation} with $C_i^u\in\mathbb{R}^{1\times n^u}$ and $C_i^s\in\mathbb{R}^{1 \times n^s}$. 

It is not difficult to verify from \eqref{eqn:Sdef} and \eqref{eqn:xi_z} that $\hat{\xi}_i(k)$ can be rewritten as
\begin{equation}\label{eqn:xi}
\hat \xi_i(k+1)=\Lambda\hat\xi_i(k)+ \1_ny_i(k+1).
\end{equation}
Since $(S^T,\beta)$ is controllable, by Lemma~\ref{lmm:regular}, for any $i\in\mathcal{V}$, we can find $G_i^u\in\mathbb{R}^{n\times n^u}$ such that
\begin{equation*}
(G_i^u)^T S^T = \left(A^u\right)^T (G_i^u)^T,\, (G_i^u)^T\beta = (C_i^uA^u)^T,
\end{equation*}
which implies that
\begin{equation}
\begin{aligned}
G_i^uA^u - \mathbf 1_n C_i^uA^u &= SG_i^u - \1_n \beta^T G_i^u \\ &= (\Lambda + 1\beta^T)G_i^u - \1_n \beta^T G_i^u=\Lambda G_i^u  ,\\
\beta^T G_i^u &= C_i^uA^u.
\end{aligned}
\end{equation}
Therefore,  we conclude that
\begin{equation}\label{eqn:G_i}
\begin{aligned}
\begin{bmatrix}
G_i^u & 0
\end{bmatrix}A-\mathbf1_nC_{i}A&=	\begin{bmatrix}
G_i^uA^u & 0
\end{bmatrix}-\1_{n}\begin{bmatrix}
C_i^uA^u & C_i^sA^s
\end{bmatrix}
\\&=\Lambda\begin{bmatrix}
G_i^u & 0
\end{bmatrix}
- \mathbf 1_{n}\begin{bmatrix}
0 &C_i^sA^s
\end{bmatrix}
,\\
\beta^T\begin{bmatrix}
G_i^u & 0
\end{bmatrix}&= \begin{bmatrix}
C_i^uA^u &0	
\end{bmatrix} = C_{i}A -
\begin{bmatrix}
0 &
C_{i}^sA^s
\end{bmatrix}.
\end{aligned}
\end{equation}

For simplicity, let us denote 
\begin{equation}\label{eqn:defGi}
G_i \triangleq \begin{bmatrix}
G_i^u & 0
\end{bmatrix}\in\mathbb{R}^{n\times n}.
\end{equation}
Moreover, define 
\begin{equation}\label{eqn:def_e}
\epsilon_i(k)\triangleq G_ix(k)-\hat\xi_i(k).
\end{equation} 
By \eqref{eqn:xi}, we thus calculate the dynamics of $\epsilon_i(k)$ as
\begin{equation}\label{eqn:e_i}
\begin{split}
&\epsilon_i(k+1) = G_ix(k+1)-\hat\xi_i(k+1)\\
&=(G_i-\1_nC_i)Ax(k)-\Lambda\hat\xi_i(k)+(G_i-\1_nC_i)w(k)\\&\quad-\1_nv_i(k+1)\\
&=(\Lambda G
-\1_{n}\begin{bmatrix}
0 &
C_{i}^sA^s
\end{bmatrix})x(k)-\Lambda\hat\xi_i(k)+(G_i-\1_nC_i)w(k)\\&\quad-\1_nv_i(k+1)\\
&=\Lambda \epsilon_i(k)-\1_{n}C_{i}^sA^sx^s(k) +(G_i-\1_nC_i)w(k)\\&\quad-\1_nv_i(k+1),
\end{split}
\end{equation}
where the third equality holds by \eqref{eqn:G_i}. It is noted that $\Lambda$ is a stable matrix and $C_{i}^sA^sx^s(k)$ is also stable. Hence, one concludes that $\cov(\epsilon_i(k))$ is bounded.

On the other hand, let us consider the dynamics of $z_i(k)$:
\begin{equation}\label{eqn:z_i}
\begin{split}
z_i(k)&=y_i(k+1)-\beta^T(G_ix(k)-\epsilon_i(k))\\
&=C_i(Ax(k)+w(k))+v_i(k+1)+\beta^T\epsilon_i(k)\\&\qquad-(C_{i}A -
\begin{bmatrix}
0 &
C_{i}^sA^s
\end{bmatrix})x(k)\\
&=\beta^T\epsilon_i(k)+C_{i}^sA^sx^s(k)+C_iw(k)+v_i(k+1).
\end{split}
\end{equation}
As previously proved, $\cov(\epsilon_i(k))$ is bounded. It thus follows that $\cov(z_i(k))$ is also bounded. 

2) To prove \eqref{eqn:localdecompose}, let us multiple both sides of \eqref{eqn:xi} by $F_i$, which gives that
\begin{equation}
F_i	\hat \xi_i(k+1)= F_i \Lambda \hat \xi_i(k)+ F_i \1_n y_i(k+1).
\end{equation}
Since $F_i$ solves \eqref{eqn:F_i}, one obtains that
\begin{equation}\label{eqn:Fi}
F_i	\hat \xi_i(k+1)= (A-KCA)F_i \hat \xi_i(k)+ K_i y_i(k+1).
\end{equation}
Summing up \eqref{eqn:Fi} for all $i=1,\cdots,m,$ yields that
\begin{equation}
\sum_{i=1}^m F_i	\hat \xi_i(k+1)= (A-KCA)\sum_{i=1}^mF_i \hat \xi_i(k)+ \sum_{i=1}^m K_i y_i(k+1).
\end{equation}
By comparing it with \eqref{eqn:optimalKF}, we complete the proof.

\section{Proof of Theorem~\ref{thm:synchronization}}\label{sec:app_synchronization}
Before proving Theorem~\ref{thm:synchronization}, we first introduce some useful preliminaries. Notice that the presence of stochastic signals $z_i(k)$ prevents us from directly applying the approaches of Lyapunov stability for deterministic systems to the analysis. We would therefore resort to a stochastic analogue of it. 

To this end, let $\{\mathcal{F}(t)\}_{t\geq 0}$ be a filtration in a
probability space $(\Omega,\mathcal{F},\mathcal{P})$ and $\{V(k)\}$ is a sequence of non-negative functions. Let us define 
\begin{equation}\label{eqn:EV}
\begin{split}
\Delta V(k) &\triangleq V(k+1)-V(k),\\
	\mathbb{E}[\Delta V(k)|\mathcal{F}(k)]&\triangleq \mathbb{E}[V(k+1)|\mathcal{F}(k)]-V(k).
\end{split}
\end{equation}
Notice that the function $V$ is defined as a supermartingale if $\mathbb{E}[\Delta V(k)|\mathcal{F}(k)] \leq 0, \;\forall k$. Moreover, if there exists $\rho>0$, such that it holds at any time that $\mathbb{E}[\Delta V(k)|\mathcal{F}(k)] \leq -\rho V(k)$, then $\mathbb{E}[V(k+1)]$ converges with an exponential rate almost surely. In order to analyze the systems which have non-zero noises at the origin, we shall extend the classical results on stability of supermartingales. Specifically,  we will consider functions that are almost supermartingales, in the sense that 
	\begin{equation}\label{eqn:cmartingale}
\mathbb{E}[\Delta V(k)|\mathcal{F}(k)] \leq -\rho V(k)+c(k),
\end{equation}
for some $\mathcal{F}(k)$-measurable random variable $c(k)$. Such functions are termed as $c$-martingale in the literature \cite{pham2009contraction,steinhardt2012finite,wang2016stochastic}. Based on their definitions, we propose $c$-martingale convergence lemma:
\begin{lemma}[$c$-martingale convergence lemma]\label{lmm:cmartingale}
Suppose there exists $\rho>0$ such that \eqref{eqn:cmartingale} holds and $\mathbb{E}[c(k)]\leq c<\infty$. Then it follows for any $k\geq 0$ that $\mathbb{E}\left[V(k)\right]$ is bounded.
\end{lemma} 
\begin{proof}
	It follows from \eqref{eqn:cmartingale} that
	\begin{equation}\label{eqn:cmartingale2}
		0\leq \mathbb{E}[V(k+1)|\mathcal{F}(k)]\leq (1-\rho)V(k)+c(k).
	\end{equation}
	By taking expectation on both sides of \eqref{eqn:cmartingale2}, it yields that
	\begin{equation}
		\begin{aligned}
			0&\leq \mathbb{E}\left[V(k+1)\right] \leq\left(1-\rho\right) \mathbb{E}\left[V(k)\right]+c\\
			&\leq \left(1-\rho\right)^{k+1} \mathbb{E}\left[V(0)\right]+c\sum_{t=0}^{k+1}\left(1-\rho\right)^t.
		\end{aligned}
	\end{equation}
The proof is thus completed.
\end{proof}

\vspace{-5pt}

Towards the proof of Theorem~\ref{thm:synchronization}, we shall respectively establish the consistency and consensus conditions.

\textit{Consistency}: By \eqref{eqn:epsilon}, we rewrite the dynamics of local state as
\begin{equation}\label{eqn:update2}
	\begin{aligned}
		\eta_i(k+1) =H\eta_i(k) &+L_iz_i(k)+BT\sum_{j=1}^m a_{ij}(\eta_{j}(k)-\eta_{i}(k))\\&+BT
		\sum_{j=1}^m a_{ij}(\epsilon_{j}(k)-\epsilon_{i}(k)).
	\end{aligned}
\end{equation}
The consistency condition is verified by summing \eqref{eqn:update2} over $i=1,\cdots,m$.

\textit{Consensus}: For simplicity, let us define the aggregated vectors and matrices as below:
\begin{equation*}
	\begin{split}
		\eta(k)&\triangleq 
		\begin{bmatrix}
			\eta_1(k)\\
			\vdots\\
			\eta_m(k)
		\end{bmatrix},\; L_\eta\triangleq 
		\begin{bmatrix}
			L_1 & & \\
			& \ddots & \\
			& & L_m
		\end{bmatrix}.\\
	\end{split}
\end{equation*}

Collecting \eqref{eqn:update2} from each sensor yields:
\begin{equation}\label{eqn:eta_matrix}
	\begin{split}
		&\eta(k+1)\\&=(I_m\otimes H)\eta(k)-[I_m\otimes (BT)](\mathcal{L}_{\mathcal{G}}\otimes I_{n(r+1)})\eta(k)\\&\qquad-[I_m\otimes (BT)](\mathcal{L}_{\mathcal{G}}\otimes I_{n(r+1)})\epsilon(k)+L_\eta z(k)\\
		&=[I_m\otimes H-\mathcal{L}_{\mathcal{G}}\otimes (BT)]\eta(k)-[\mathcal{L}_{\mathcal{G}}\otimes (BT)]\epsilon(k)\\&\qquad+L_\eta z(k),
	\end{split}
\end{equation}
where $z(k)$ is defined in \eqref{eqn:S}. Let us rewrite the average state of all sensors as  
\begin{equation}
	\bar{\eta}(k)= \frac{1}{m} \sum_{i=1}^m \eta_{i}(k)=\frac{1}{m} (\1_m^T\otimes I_{n(r+1)})\eta(k).
\end{equation}
As $\1_m^T \mathcal{L}_{\mathcal{G}} =0$, it follows that
\begin{equation*}\label{eqn:average_matrix}
	\begin{split}
		\bar{\eta}(k+1) &= \frac{1}{m} (\1_m^T\otimes I_{n(r+1)})\Big([I_m\otimes H-\mathcal{L}_{\mathcal{G}}\otimes (BT)]\eta(k)\\&\quad-[\mathcal{L}_{\mathcal{G}}\otimes ( BT)]\epsilon(k)+L_\eta z(k)\Big)\\&=H\bar{\eta}(k)+\frac{1}{m}(\1_m^T\otimes I_{n(r+1)})L_\eta z(k).
	\end{split}
\end{equation*}

\vspace{0pt}

Furthermore, we define for each sensor $i$ that $$\delta_i(k)\triangleq \eta_i(k)-\bar\eta(k).$$ By stacking $\delta_i(k)$ together, let us denote
\begin{equation}
	\delta(k)\triangleq \col(\delta_1(k),\cdots,\delta_m(k)).
\end{equation}
We therefore have
\begin{equation}\label{eqn:error_matrix}
	\begin{split}
		\delta&(k+1)=[I_m\otimes H-\mathcal{L}_{\mathcal{G}}\otimes (BT)]\delta(k)\\&+[(I_{m}-\frac{1}{m}\1_m\1_m^T)\otimes I_{n(r+1)}]L_\eta z(k)-[\mathcal{L}_{\mathcal{G}}\otimes ( BT)]\epsilon(k).
	\end{split}
\end{equation}

\vspace{0pt}

By \cite{you2011network}, there always exists a unitary matrix $$\Phi\triangleq[\frac{1}{\sqrt{m}}\1_m,\phi_2,\cdots,\phi_m],$$ with which the Laplacian matrix can be diagonalized as
$$
\Phi^T\mathcal{L}_{\mathcal{G}}\Phi=\diag(0,\mu_2,\cdots,\mu_m).
$$
One hence concludes
\begin{equation}
	\begin{split}
		&(\Phi \otimes I_{n(r+1)})^T[\mathcal{L}_{\mathcal{G}}\otimes (BT)](\Phi \otimes I_{n(r+1)})
		\\&=\diag(0,\mu_2BT,...,\mu_mBT),\\
		&(\Phi \otimes I_{n(r+1)})^T[I_m\otimes H-\mathcal{L}_{\mathcal{G}}\otimes (B\Gamma)](\Phi \otimes I_{n(r+1)})
		\\&=\diag(H,H-\mu_2BT,..., H-\mu_mBT),
	\end{split}
\end{equation}
which holds by the property of Kronecker product. Denote 
\begin{equation}\label{eqn:delta}
	\begin{split}
		\tilde\delta(k)\triangleq(\Phi\otimes I_{n(r+1)})^T\delta(k),\;
		\tilde\epsilon(k)\triangleq(\Phi\otimes I_{n(r+1)})^T\epsilon(k).
	\end{split}
\end{equation}
Let us further partition $\tilde\delta(k)$ into two parts, i.e., $\tilde\delta(k)=[\tilde\delta^T_1(k),\tilde\delta^T_2(k)]^T$, where $\tilde\delta_1(k)\in\mathbb{R}^{n(r+1)}$ consists of the first $n(r+1)$ entries of $\tilde\delta(k)$. One thus obtains from \eqref{eqn:error_matrix} that
\begin{equation}\label{eqn:trajectory}
	\begin{split}
		\tilde\delta_1(k+1) &= \frac{1}{\sqrt{m}}\sum_{i=1}^m \delta_i(k+1)=0,\\
		\tilde\delta_2(k+1)&=A_{\delta} \tilde\delta_2(k)+L_{z} z(k)+B_{\epsilon} \tilde\epsilon_2(k),
	\end{split}
\end{equation} 
where $A_{\delta} \triangleq \diag(H-\mu_2BT,\cdots, H-\mu_mBT)$, $B_{\epsilon}  \triangleq \diag(-\mu_2 BT,\cdots,-\mu_mBT)$, and $L_{z}$ is formed by the last $mn(r+1)-n(r+1)$ rows of $[(\Phi^T-\frac{1}{m}\Phi^T\1_m\1_m^T)\otimes I_{n(r+1)}]L_\eta$. 


Clearly, $\tilde\delta_1(k+1)$ is stable. We thus focus on the stability of $\tilde\delta_2(k+1)$. To proceed, let us consider the following Lyapunov candidate:
\begin{equation}
	V(k)=\tilde\delta^T_2(k) \mathcal P \tilde\delta_2(k).
\end{equation}
It thus follows that
\begin{equation}\label{eqn:trace}
	\begin{aligned}
		V(k)&=\tr(\tilde\delta^T_2(k) \mathcal P \tilde\delta_2(k)) =\tr(\mathcal P\tilde\delta_2(k) \tilde\delta^T_2(k))\\&\leq \tr(\mathcal P)\tr(\tilde\delta_2(k) \tilde\delta^T_2(k))=\tr(\mathcal P)||\tilde{\delta}_2(k)||^2,
	\end{aligned}
\end{equation}
where the ineqaulity holds since 
$
\tr(\mathcal{A} \mathcal B) \leq \tr(\mathcal A) \tr(\mathcal B)
$ holds for any $\mathcal{A}, \mathcal{B}\geq 0$.
(\hspace{1pt}\cite{coope1994matrix}).
The difference of $V(k)$ along \eqref{eqn:trajectory} is given by
\begin{equation}
	\begin{split}
		&\mathbb{E}[\Delta V(k)|\mathcal{F}(k)]\triangleq \mathbb{E}[V(k+1)-V(k)|\mathcal{F}(k)]\\=&\;\tilde\delta^T_2(k)(A_{\delta}^T \mathcal P A_{\delta}-\mathcal{P})\tilde\delta_2(k)+2\tilde\delta^T_2(k) A_{\delta}^T \mathcal P L_z \mathbb{E}[z(k)|\mathcal{F}(k)]\\&\qquad+2\tilde\delta^T_2(k) A_{\delta}^T \mathcal P B_\delta\tilde{\epsilon}(k)+2\mathbb{E}[z^T(k)|\mathcal{F}(k)] L_z^T \mathcal P B_\delta\tilde{\epsilon}(k)\\&\qquad+L_{z}^T \mathcal PL_{z}\mathbb{E}[z^T(k)z(k)|\mathcal{F}(k)]+\tilde{\epsilon}^T(k) B_{\epsilon}^T \mathcal P B_\epsilon \tilde{\epsilon}(k).
	\end{split}
\end{equation}
In view of Lemma~\ref{lmm:eta}, $A_{\delta}$ is stable. Hence, there exist $\mathcal Q >0$ and $\sigma_1,\sigma_2>0 $ such that
\begin{equation}
	(1+\sigma_1+\sigma_2)A_{\delta}^T \mathcal P A_{\delta}-\mathcal{P} +\mathcal Q =0.
\end{equation}
Now using Young's inequality, one concludes that
\begin{equation*}\label{eqn:Lyapunov}
	\begin{split}
		&\quad\; \mathbb{E}[\Delta V(k)|\mathcal{F}(k)]\\&\leq\tilde\delta^T_2(k)[(1+\sigma_1+\sigma_2)A_{\delta}^T \mathcal P A_{\delta}-\mathcal{P}]\tilde\delta_2(k)\\&\quad+(1+\sigma_2^{-1}+\sigma_3)L_{z}^T \mathcal PL_{z} \mathbb{E}[z^T(k)z(k)|\mathcal{F}(k)]\\&\quad+(1+\sigma_1^{-1}+\sigma_3^{-1})\tilde{\epsilon}^T(k) B_{\epsilon}^T \mathcal P B_\epsilon \tilde{\epsilon}(k)\\&\leq -\lambda_{\min}(\mathcal Q)||\tilde{\delta}_2(k)||^2+c(k)\\&\leq -\frac{\lambda_{\min}(\mathcal Q)}{\tr(\mathcal P)}V(k)+c(k),
	\end{split}
\end{equation*}
where the last inequality holds by \eqref{eqn:trace}. As proved in Lemma~\ref{lmm:localeqv}, $\cov(z_i(k))$ is bounded at any time. Moreover, $||\epsilon(k)||^2$ is also bounded by \eqref{eqn:hbar}. It thus follows that
$
	\mathbb{E}[c(k)]<\infty.
$
In view of Lemma~\ref{lmm:cmartingale}, we conclude that $\mathbb{E}[V(k)]$ is bounded. 
As a result of \eqref{eqn:trace}, $\cov(\tilde\delta_2(k))$ is also bounded. Combining it with \eqref{eqn:delta}, it follows that $\cov[\eta_i(k)-\bar{\eta}(k)]$ is bounded for any $i$, which completes the proof.

\bibliographystyle{IEEEtran}
\bibliography{referenceKF}

 \end{document}